\documentclass[11pt,english,letter]{article}
\usepackage[T1]{fontenc}
\usepackage[latin1]{inputenc}
\usepackage{amsmath}
\usepackage{graphicx}
\usepackage{amssymb}

\makeatletter

\usepackage{amsfonts}

\usepackage{epsfig}

\usepackage{mathrsfs}

\usepackage[scriptsize]{subfigure}
\providecommand{\U}[1]{\protect\rule{.1in}{.1in}}
\newcounter{algo}
\newcounter{rem}
\newenvironment{algo}[2]{\refstepcounter{algo}\label{#2}   \begin{center}
\begin{minipage}{0.9\textwidth}   \hrule\smallskip
\textbf{Algorithm \thealgo: #1}
\par\smallskip\hrule\smallskip\ignorespaces}{\par\smallskip\hrule
\end{minipage}
\end{center}
}
\newtheorem{theorem}{Theorem}
\newtheorem{proposition}{Proposition}
\newtheorem{definition}{Definition}
\newtheorem{lemma}{Lemma}
\newtheorem{corollary}{Corollary}

\newenvironment{proof}[1][Proof]{\noindent \textbf{#1.} }{\qedsymbol}
\newcommand{\qedsymbol}{\hspace{\fill}\rule{1.5ex}{1.5ex}}
\def\baselinestretch{1.30}
\usepackage{babel}
\makeatother

\begin{document}

\title{\vspace{-1.5cm}
{\Huge Asynchronous Iterative Waterfilling for Gaussian }{\normalsize{} }{\Huge Frequency-Selective
Interference Channels}}

\author{Gesualdo Scutari$^{1}$, Daniel P. Palomar$^{2}$, and Sergio Barbarossa$^{1}$
\\
 %EndAName
{\small E-mail: $\{$scutari, sergio$\}$@infocom.uniroma1.it, palomar@ust.hk}{\normalsize }\\
{\normalsize{} $^{1\text{ }}$}{\small Dpt. INFOCOM, Univ. of Rome
{}``La Sapienza\textquotedblright, Via Eudossiana 18, 00184 Rome,
Italy.}{\normalsize{} }\\
{\normalsize{} $^{2}$ }{\small Dpt. of Electronic and Computer Eng.,
Hong Kong Univ. of Science and Technology, Hong Kong.}}

\date{{\small Submitted to IEEE }\textit{\small Transactions on Information
Theory}{\small , August 22, 2006.}{\normalsize }\\
{\normalsize{} }{\small Revised September 25, 2007. Accepted January
14, 2008.\thanks{Part of this work was presented in \textit{IEEE Workshop
on Signal Processing Advances in Wireless Communications, (SPAWC-2006)},
July 2-5, 2006, and in \textit{Information Theory and Applications
(ITA) Workshop}, Jan. 29 - Feb. 2, 2007. This work was supported by
the SURFACE project funded by the European Community under Contract
IST-4-027187-STP-SURFACE.}}}

\maketitle
\vspace{-1cm}

\begin{abstract}
This paper considers the maximization of information rates for the
Gaussian frequency-selective interference channel, subject to power
and spectral mask constraints on each link. To derive decentralized
solutions that do not require any cooperation among the users, the
optimization problem is formulated as a static noncooperative game
of complete information. To achieve the so-called Nash equilibria
of the game, we propose a new distributed algorithm called asynchronous
iterative waterfilling algorithm. In this algorithm, the users update
their power spectral density in a completely distributed and asynchronous
way: some users may update their power allocation more frequently
than others and they may even use outdated measurements of the received
interference. The proposed algorithm represents a unified framework
that encompasses and generalizes all known iterative waterfilling
algorithms, e.g., sequential and simultaneous versions. 
The main result of the paper consists of a unified set of conditions that guarantee the
global converge of the proposed algorithm to the (unique) Nash equilibrium
of the game. 
\bigskip{}

\end{abstract}
\vspace{-0.6cm}

\section{Introduction and Motivation}

In this paper we focus on the frequency selective interference channel
with Gaussian noise. %The interference channel is a relevant
%mathematical model of many communication systems such as wireless
%ad-hoc networks \cite{Akyildiz-Wang, Goldsmith-Wicker} or Digital
%Subscriber Lines (DSL) \cite{Starr-CioffiBook, Yu, ChungISIT03},
%where the distributed nature of the transmit-receive pairs does not
%allow cooperation at the transmission level.
The capacity region of the interference channel is still unknown.
Only some bounds are available (see, e.g., \cite{Han-Kobayashi,Meulen}
for a summary of the known results about the Gaussian interference
channel). A pragmatic approach that leads to an achievable region
or inner bound of the capacity region is to restrict the system to
operate as a set of independent units, i.e., not allowing multiuser
encoding/decoding or the use of interference cancelation techniques.
This achievable region is very relevant in practical systems with
limitations on the decoder complexity and simplicity of the system.
With this assumption, multiuser interference is treated as noise and
the transmission strategy for each user is simply his power allocation.
The system design reduces then to finding the optimum Power Spectral
Density (PSD) for each user, according to a specified performance
metric.

Within this context, existing works \cite{Cendrillon-Yu}$-$\cite{Tse}
considered the maximization of the information rates of all the links,
subject to transmit power and (possibly) spectral mask constraints
on each link. The latter constraints are especially motivated in adaptive
scenarios,  e.g., cognitive radio, where previously allocated spectral
bands may be reused, but provided that the generated interference
falls below specified thresholds \cite{Haykin}. In \cite{Cendrillon-Yu,Yu-Lui},
a centralized approach based on duality theory  was proposed
to compute, under technical conditions, the largest achievable rate
region of the system (i.e., the Pareto-optimal set of the achievable
rates). Our interest, in this paper, is focused on finding distributed
algorithms with no centralized control and no cooperation among the
users. Hence, we cast the system design under the convenient framework
of game theory. In particular, we formulate the rate maximization
problem as a strategic non-cooperative game of complete information,
where every link is a player that competes against the others by choosing
the spectral power allocation that maximizes his own information rate.
%This
%changes the original multi-objective optimization problem into a set
%of single-objective optimization problems coupled together (this
%coupling is what makes the problem hard to solve).
An equilibrium for the whole system is reached when every player is
unilaterally optimum, i.e., when, given the current strategies of
the others, any change in his own strategy would result in a rate
loss. This equilibrium constitutes the celebrated notion of Nash Equilibrium
(NE) \cite{Nash-paper}.

The Nash equilibria of the rate maximization game can be reached using
Gaussian signaling and a proper PSD from each user \cite{Scutari-Palomar-Barbarossa-ISIT}$-$\cite{Scutari_Thesis}.
To obtain the optimal PSD of the users, Yu, Ginis, and Cioffi proposed
the \textit{sequential} Iterative WaterFilling Algorithm (IWFA) \cite{Yu}
in the context of DSL systems, modeled as a Gaussian frequency-selective
interference channel. The algorithm is an instance of the Gauss-Seidel
scheme \cite{Bertsekas Book-Parallel-Comp}: the users maximize their
own information rates \emph{sequentially} (one after the other), according
to a fixed updating order. Each user performs the single-user waterfilling
solution given the interference generated by the others as additive
(colored) noise. The most appealing features of the sequential IWFA
are its low-complexity and distributed nature. In fact, to compute
the waterfilling solution, each user only needs to measure the noise-plus-interference
PSD, without requiring specific knowledge of the power allocations
and the channel transfer functions of all other users. The convergence
of the sequential IWFA has been studied in a number of works \cite{ChungISIT03},
\cite{Scutari-Barbarossa-GT_PII}$-$\cite{Tse}, each time obtaining
milder convergence conditions. However, despite its appealing properties,
the sequential IWFA may suffer from slow convergence if the number
of users in the network is large, just because of the sequential updating
strategy. In addition, the algorithm requires some form of central
scheduling to determine the order in which users update their PSD.

To overcome the drawback of slow speed of convergence, the \textit{simultaneous}
IWFA was proposed in \cite{Scutari-Palomar-Barbarossa-ISIT,Scutari-Barbarossa-GT_PII,Scutari_Thesis}.
The simultaneous IWFA is an instance of the Jacobi scheme \cite{Bertsekas Book-Parallel-Comp}:
at each iteration, the users update their own strategies \emph{simultaneously},
still according to the waterfilling solution, but using the interference
generated by the others in the \emph{previous} iteration. The simultaneous
IWFA was shown to converge to the unique NE of the rate maximization
game faster than the sequential IWFA and under weaker conditions on
the multiuser interference than those given in \cite{Yu,ChungISIT03},
\cite{Yamashitay-Luo}$-$\cite{Tse} for the sequential IWFA. Furthermore,
differently from \cite{Yu,ChungISIT03, Yamashitay-Luo, Tse},
the algorithm as proposed in \cite{Scutari-Barbarossa-GT_PII} takes explicitly into account the spectral masks constraints.
However, the simultaneous IWFA still requires some form of synchronization,
as all the users need to be updated simultaneously. Clearly, in a
real network with many users, the synchronization requirement of both
sequential and simultaneous IWFAs goes against the non-cooperation
principle and it might be unacceptable. %Ideally, we would like to
%show convergence of a more general IWFA where users update their
%PSDs at different and arbitrary times in a non-systematic fashion.
%In addition, we would also like to consider the fact that the
%updates of the users may use outdated system information. This is
%especially relevant in wireless sensor networks where, because of
%limited energy capabilities, some nodes are allowed to fall asleep,
%to recharge their batteries, at random times, for a priori unknown
%time intervals.

This paper generalizes the existing results for the sequential and
simultaneous IWFAs and develops a unified framework based on the so-called
\emph{asynchronous} IWFA, that falls within the class of totally asynchronous
schemes of \cite{Bertsekas Book-Parallel-Comp}. In this more general
algorithm, all users still update their power allocations according
to the waterfilling solution, but the updates can be performed in
a \emph{totally asynchronous} way (in the sense of \cite{Bertsekas Book-Parallel-Comp}).
This means that some users may update their power allocations \emph{more
frequently} than others and they may even use an \emph{outdated} measurement
of the interference caused from the others. These features make the
asynchronous IWFA appealing for all practical scenarios, either wired
or wireless, as it strongly relaxes the need for coordinating the
users' updating schedule.

The main contribution of this paper is to derive sufficient conditions
for the global convergence of the asynchronous IWFA to the (unique)
NE of the rate maximization game. %Our convergence conditions enlarge
%those given in \cite{Yu, ChungISIT03}, \cite{Yamashitay-Luo}$-$
%\cite{Tse} for the convergence of the sequential IWFA and,
Interestingly, our convergence conditions are shown to be independent
of the users' update schedule. Hence, they represent a unified set
of conditions encompassing all existing algorithms, either synchronous
or asynchronous, that can be seen as special cases of our asynchronous
IWFA. Our conditions also imply that both sequential and simultaneous
algorithms are robust to situations where some users may fail to follow
their updating schedule. Finally, we show that our sufficient conditions
for the convergence of the asynchronous IWFA coincide with those given
recently in \cite{Scutari-Barbarossa-GT_PII} for the convergence
of the (synchronous) sequential and simultaneous IWFAs, and are larger than  
conditions obtained in \cite{Yu,ChungISIT03}, \cite{Yamashitay-Luo}$-$\cite{Tse} for the convergence of the sequential IWFA in the absence of spectral mask constraints.

The paper is organized as follows. Section \ref{Sec:System.Model}
provides the system model and formulates the optimization problem
as a strategic non-cooperative game. Section \ref{Sec:AIWFA} contains
the main result of the paper: the description of the proposed asynchronous
IWFA along with its convergence properties. Section \ref{Sec:two examples}
recovers the sequential and simultaneous IWFAs as special cases of
the asynchronous IWFA and then, as a by product, it provides a unified
set of convergence conditions for both algorithms. Finally, Section \ref{Sec:Conclusion}
draws some conclusions.

\vspace{-0.3cm}

\section{System Model and Problem Formulation}

\label{Sec:System.Model} In this section we clarify the assumptions
and the constraints underlying the system model and we formulate the
optimization problem addressed in this paper explicitly. \vspace{-0.3cm}

\subsection{System model}

\label{Subsec:system_model} \vspace{-0.3cm}
 We consider a Gaussian frequency-selective interference channel composed
by multiple links. Since our goal is to find distributed algorithms
that require neither a centralized control nor a coordination
among the users, we focus on transmission techniques where interference
cancelation is not possible and multiuser interference is treated
by each receiver as additive colored noise. The channel frequency-selectivity
is handled, with no loss of optimality, adopting a multicarrier 
transmission strategy.%
\footnote{Multicarrier transmission is a capacity-lossless strategy for sufficiently large
block length \cite{Cover,Tse-book}.%
} Given the above system model, we make the following assumptions:

\noindent \textbf{A.1} Each channel changes sufficiently slowly to
be considered fixed during the whole transmission, so that the information
theoretic results are meaningful;

\noindent \textbf{A.2} The channel from each source to its own destination
is known to the intended receiver, but not to the other terminals;
each receiver is also assumed to measure with no errors the overall  PSD of
the noise plus interferences generated by the other users.
Based on this information, each receiver computes the optimal power
allocation across the frequency bins for its own transmitter and feeds
it back to its transmitter through a low bit rate (error-free) feedback
channel.%
\footnote{In practice, both measurements and feedback are inevitably affected
by errors. This scenario can be studied by extending our formulation
to games with partial information \cite{Osborne,Aubin-book}, but
this goes beyond the scope of the present paper.%
}

\noindent \textbf{A.3} All the users are block-synchronized
with an uncertainty at most equal to the cyclic prefix length. This
imposes a minimum length of the cyclic prefix that will depend on
the maximum channel delay spread. %spread in the system. Moreover, frequency offsets among different transmitters
%and receivers are assumed sufficiently small so that each frequency bin
%operates independently and free from intercarrier interference. We will relax
%these assumption in Section \ref{Synch-error}, where we will explicitly take
%into account the effect of time and/or frequency offsets.

We consider the following constraints:

\noindent \textbf{Co.1} Maximum overall transmit power for each user:\vspace{-0.1cm}
 \begin{equation}
E\left\{ \left\Vert \mathbf{s}_{q}\right\Vert _{2}^{2}\right\} =\sum_{k=1}^{N}\bar{p}_{q}(k)\leq NP_{q},\quad\quad  q=1,\ldots,Q,\vspace{-0.1cm}\label{power-constraint}\end{equation}
where $\mathbf{s}_{q}$ contains the $N$ symbols transmitted by
user $q$ on the $N$ carriers, $\bar{p}_{q}(k)\triangleq E\left\{ \left\vert s_{q}(k)\right\vert ^{2}\right\} $
denotes the power allocated by user $q$ over carrier $k,$ and $P_{q}$
is power in units of energy per transmitted symbol. %The symbols are assumed
%to be, without loss of generality (w.l.o.g.), zero-mean unit energy
%uncorrelated symbols, i.e., $E\left\{
%\mathbf{s}_{q}\mathbf{s}_{q}^{H}\right\}  =\mathbf{I}_{N}$.

\noindent \textbf{Co.2} Spectral mask constraints: \begin{equation}
E\left\{ \left\vert s_{q}(k)\right\vert ^{2}\right\} =\bar{p}_{q}(k)\leq\bar{p}_{q}^{\max}(k),\quad\quad  k=1,\ldots,N,\quad q=1,\ldots,Q,\label{mask_on_F_q_}\end{equation}
 where $\bar{p}_{q}^{\max}(k)$ represents the maximum power that
is allowed to be allocated on the $k$-th frequency bin from the $q$-th
user. Constraints like (\ref{mask_on_F_q_}) are imposed to limit
the amount of interference generated by each transmitter over pre-specified
bands.

The main goal of this paper is to obtain the optimal vector power
allocation $\mathbf{p}_{q}\triangleq (p_{q}(1),\ldots,p_{q}(N))$, for each
user, according to the optimality criterion introduced in the next
section.\vspace{-0.2cm}

\subsection{Problem formulation as a game}

\vspace{-0.1cm}
 \label{Sec:Problem-Formulation} %We formulate the \emph{joint
%}maximization of mutual information on each link within the framework of game
%theory (see, e.g., \cite{Osborne, Aubin-book}), using as optimality criterion the concept
%of Nash Equilibrium (NE) \cite{Nash-paper}. Specifically,

We consider a strategic non-cooperative game \cite{Osborne,Aubin-book},
in which the players are the links and the payoff functions are the
information rates on each link: Each player competes rationally%
\footnote{The rationality assumption means that each user will never chose a
strictly dominated strategy. A strategy profile $\mathbf{x}_{q}$
is strictly dominated by $\mathbf{z}_{q}$ if $\Phi_{q}\left(\mathbf{x}_{q},\mathbf{y}_{-q}\right)<\Phi_{q}\left(\mathbf{z}_{q},\mathbf{y}_{-q}\right),$
for a given admissible $\mathbf{y}_{-q}\triangleq\left(\mathbf{y}_{1},\ldots,\mathbf{y}_{q-1},\mathbf{y}_{q+1},\ldots,\mathbf{y}_{Q}\right),$
where $\Phi_{q}$ denotes the payoff function of player $q.$%
} against the others by choosing the strategy that maximizes his own
rate, given constraints \textbf{Co.1} and \textbf{Co.2}. A NE of the
game is reached when every user, given the strategy profile of the
others, does not get any rate increase by changing his own strategy.

Using the signal model described in Section \ref{Subsec:system_model},
the achievable rate for each player $q$ is computed as the maximum
information rate on the $q$-th link, assuming \emph{all the other
received signals as additive colored noise}. It is straightforward
to see that a (pure or mixed strategy) NE is obtained if each user
transmits using Gaussian signaling, with a proper PSD. In fact, for
each user, given that all other users use Gaussian codebooks, the
optimal codebook maximizing mutual information is also Gaussian \cite{Cover}.%
\footnote{Observe that, in general, Nash equilibria achievable using arbitrary
non-Gaussian codes may exist. In this paper, we focus only on transmission
using  Gaussian codebooks.%
} Hence, the maximum achievable rate for the $q$-th user is given
by \cite{Cover}\begin{equation}
R_{q}=\frac{1}{N}\sum_{k=1}^{N}\log\left(1+\mathsf{sinr}_{q}(k)\right),\label{Rate}\end{equation}
 with $\mathsf{sinr}_{q}(k)$ denoting the Signal-to-Interference
plus Noise Ratio (SINR) on the $k$-th carrier for the $q$-th link:
\begin{equation}
\mathsf{sinr}_{q}(k)\triangleq\frac{\left\vert \bar{H}_{qq}(k)\right\vert ^{2}\bar{p}_{q}(k)/d_{qq}^{\gamma}}{\sigma_{{q}}^{2}(k)+\sum_{\, r\neq q}\left\vert \bar{H}_{qr}(k)\right\vert ^{2}\bar{p}_{r}(k)/d_{rq}^{\gamma}}\triangleq\frac{\left\vert H_{qq}(k)\right\vert ^{2}p_{q}(k)}{\sigma_{q}^{2}(k)+\sum_{\, r\neq q}\left\vert H_{qr}(k)\right\vert ^{2}p_{r}(k)},\label{SINR_q}\end{equation}
 where %\begin{equation}
%H_{rq}(k)%
%%TCIMACRO{\TeXButton{def}{\triangleq}}%
%%BeginExpansion
%\triangleq
%%EndExpansion
%\bar{H}_{rq}(k)\sqrt{\dfrac{P_{r}}{d_{rq}^{\gamma}}}, \label{normalized_Hrq}%
%\end{equation}
$\bar{H}_{qr}(k)$ denotes the frequency-response of the channel between
source $r$ and destination $q$ excluding the path-loss $d_{qr}^{\gamma}$
with exponent $\gamma$ and $d_{qr}$ is the distance between source
$r$ and destination $q$; $\sigma_{{q}}^{2}(k)$ is the variance
of the zero mean circularly symmetric complex Gaussian noise at receiver
$q$ over the carrier $k$; and for the convenience of notation, we
have introduced the normalized quantities $H_{qr}(k)\triangleq\bar{H}_{qr}(k)\sqrt{P_{r}/d_{qr}^{\gamma}}$
and $p_{q}(k)\triangleq\bar{p}_{q}(k)/P_{q}\mathbf{\ }$.%denotes the
%frequency-response of the channel between source $r$ and destination $q$,
%including the path-loss $d_{qr}^{\gamma}$ with exponent $\gamma$ and
%normalized fading $\bar{H}_{qr}(k)$;
% and $p_{q}(k)%
%%TCIMACRO{\TeXButton{def}{\triangleq}}%
%%BeginExpansion
%\triangleq
%%EndExpansion%
%%TCIMACRO{\TeXButton{bar_{p}}{\bar{p}}}%
%%BeginExpansion
%\bar{p}%
%%EndExpansion
%_{q}(k)/P_{q}\mathbf{\ }$is the normalized power allocated by the $q$-th user
%over the $k$-th carrier.

Observe that in the case of practical coding schemes, where only finite
order constellations can be used, we can use the gap approximation
analysis \cite{Forney,Goldsmith-Chua} and write the number of bits
transmitted over the $N$ substreams from the $q$-th source still
as in (\ref{Rate}) but  replacing $\left\vert H_{qq}(k)\right\vert ^{2}$
in (\ref{SINR_q}) with $\left\vert H_{qq}(k)\right\vert ^{2}/\Gamma_{q},$where
$\Gamma_{q}\geq1$ is the gap. The gap depends only on the family
of constellation and on $P_{e,q}$; for $M$-QAM constellations, for
example, if the symbol error probability \ is approximated by $P_{e,q}(\mathsf{sinr}_{q}(k))\approx4\mathcal{Q}\left(\sqrt{3\,\mathsf{sinr}_{q}(k)/(M-1)}\right),$
the resulting gap is $\Gamma_{q}=(\mathcal{Q}^{-1}(P_{e,q}/4))^{2}/3$
\cite{Forney,Goldsmith-Chua}. %\cite{Palomar-Barbarossa}.

In summary, we have a game with the following structure: \begin{equation}
{\mathscr{G}}=\left\{ \Omega,\{{\mathscr{P}}_{q}\}_{q\in\Omega},\{{R}_{q}\}_{q\in\Omega}\right\} ,\label{Game G}\end{equation}
 where $\ \Omega\triangleq\left\{ 1,2,\ldots,Q\right\} $ denotes
the set of the $Q$ active links, ${\mathscr{P}}_{q}$ is the set
of admissible (normalized) power allocation strategies, across the
$N$ available carriers, for the $q$-th player, defined as%
\footnote{In order to avoid the trivial solution $p_{q}^{\star}(k)=p_{q}^{\max}(k)$
for all $k$, $\sum_{k=1}^{N}p_{q}^{\max}({k})>N$ is assumed for
all $q\in\Omega$. Furthermore, in the feasible strategy set of each
player, we can replace, without loss of generality, the original \textit{inequality}
power constraint $(1/N)\ \sum_{k=1}^{N}p_{q}(k)\leq1,$ with equality$,$
since, at the optimum, this constraint must be satisfied with equality.%
}\begin{equation}
\hspace{-0.2cm}{\mathscr{P}}_{q}\!\triangleq\!\left\{ \mathbf{p}_{q}\!\!\in\!\!\mathcal{\ \mathbb{R}}^{N}:\frac{1}{N}\!\sum_{k=1}^{N}p_{q}(k)=1,\quad0\leq p_{q}(k)\leq p_{q}^{\max}(k),\quad k=1,\ldots,N\right\} ,\label{admissible strategy set}\end{equation}
 with $p_{q}^{\max}(k)\triangleq\overline{p}_{q}^{\max}(k)/P_{q}$
and ${R}_{q}$ is the payoff function of the $q$-th player, defined
in (\ref{Rate}).

The optimal strategy for the $q$-th player, given the power allocation
of the others, is then the solution to the following maximization
problem%
\footnote{In the optimization problem in (\ref{Rate Game}), any concave strictly
increasing function of the rate can be equivalently considered as
payoff function of each player. The optimal solutions of this new
set of problems coincide with those of (\ref{Rate Game}).%
} \begin{equation}
\begin{array}{l}
\operatorname*{maximize}\limits _{\mathbf{p}_{q}}\quad\ \dfrac{1}{N}{\displaystyle \sum\limits _{k=1}^{N}}\log\left(1+\text{ }\mathsf{sinr}_{q}(k)\right)\\
\operatorname*{subject}\text{ }\operatorname*{to}\text{\ \ \ }\mathbf{p}_{q}\in{\mathscr{P}}_{q}\end{array},\qquad\qquad\forall q\in\Omega\label{Rate Game}\end{equation}
 where $\mathsf{sinr}_{q}(k)$ and$\ {\mathscr{P}}_{q}$\ and are
given in (\ref{SINR_q}) and (\ref{admissible strategy set}), respectively.
Note that, for each $q$, the maximum in (\ref{Rate Game}) is taken
over $\mathbf{p}_{q},$ for a \textit{fixed} $\mathbf{p}_{-q}\triangleq\left(\mathbf{p}_{1},\ldots,\mathbf{p}_{q-1},\mathbf{p}_{q+1},\ldots,\mathbf{p}_{Q}\right).$

The solutions of (\ref{Rate Game}) are the well-known Nash Equilibria,
which are formally defined as follows. \begin{definition} \label{NE def}
A (pure) strategy profile $\mathbf{p}^{\star}=\left(\mathbf{p}_{1}^{\ast},\ldots,\mathbf{p}_{Q}^{\ast}\right)\in{\mathscr{P}}_{1}\times\ldots\times{\mathscr{P}}_{Q}$
\ is a Nash Equilibrium of the game ${\mathscr{G}}$ in $($\ref{Game G}$)$
if\vspace{-0.3cm}
 \begin{equation}
R_{q}(\mathbf{p}_{q}^{\star},\mathbf{p}_{-q}^{\star})\geq R_{q}(\mathbf{p}_{q},\mathbf{p}_{-q}^{\star}),\ \text{\ \ }\forall\mathbf{p}_{q}\in{\mathscr{P}}_{q},\text{ }\forall q\in\Omega.\label{pure-NE}\end{equation}
 \end{definition} %The definition of NE as given in (\ref{pure-NE}) can be generalized to contain

Observe that, for the payoff functions defined in (\ref{Rate}), we
can indeed limit ourselves to adopt pure strategies w.l.o.g., as we
did in (\ref{Game G}), since every NE of the game is proved to be
achievable using pure strategies in \cite{Scutari-Barbarossa-GT_PI}.

According to (\ref{Rate Game}), all the (pure) Nash equilibria of the game, if
they exist, must satisfy the waterfilling solution \emph{for each}
user, i.e., the following system of \emph{nonlinear} equations: \begin{equation}
\begin{array}{c}
\mathbf{p}_{q}^{\star}=\mathsf{WF}_{q}\left(\mathbf{p}_{1}^{\star},\ldots,\mathbf{p}_{q-1}^{\star},\mathbf{p}_{q+1}^{\star},\ldots,\mathbf{p}_{Q}^{\star}\right)=\mathsf{WF}_{q}(\mathbf{p}_{-q}^{\star})\end{array},\quad\forall q\in\Omega,\label{sym_WF-sistem_mask}\end{equation}
 with the waterfilling operator $\mathsf{WF}_{q}\left(\mathbf{\cdot}\right)$
defined as \begin{equation}
\left[\mathsf{WF}_{q}\left(\mathbf{p}_{-q}\right)\right]_{k}\triangleq\left[\mu_{q}-\dfrac{\sigma_{{q}}^{2}(k)+\sum_{\, r\neq q}\left\vert H_{qr}(k)\right\vert ^{2}p_{r}(k)}{\left\vert H_{qq}(k)\right\vert ^{2}}\right]_{0}^{p_{q}^{\max}(k)},\quad k=1,\ldots,N,\label{WF_mask}\end{equation}
 where $\left[x\right]_{a}^{b}$ denotes the Euclidean projection
of $x$ onto the interval $[a,b]$.%
\footnote{The Euclidean projection $\left[x\right]_{a}^{b}$, with $a\leq b$, is defined as
follows: $\left[x\right]_{a}^{b}=a$, if $x\leq a$, $\left[x\right]_{a}^{b}=x$,
if $a<x<b$, and $\left[x\right]_{a}^{b}=b$, if $x\geq b$.%
} The water-level $\mu_{q}$ is chosen to satisfy the power constraint
$(1/N)\sum_{k=1}^{N}p_{q}^{\star}(k)=1.$

Observe that in the absence of spectral mask constraints (i.e., when
$p_{q}^{\max}(k)=+\infty$, $\forall q,$ $\forall k$), the Nash
equilibria of game ${\mathscr{G}}$ are given by the classical simultaneous
waterfilling solutions \cite{Yu,ChungISIT03}, where $\mathsf{WF}_{q}\left(\mathbf{\cdot}\right)$
in (\ref{sym_WF-sistem_mask}) is still obtained from (\ref{WF_mask})
simply setting $p_{q}^{\max}(k)=+\infty$, $\forall q,$ $\forall k.$
%In this special case, the game ${\mathscr{G}}$ in (\ref{Game G})
%is usually referred to in the literature as the Gaussian Interference
%Game \cite{Yu,ChungISIT03}, and alternative (sufficient) conditions
%for the existence and uniqueness of a NE are given in \cite{Yu,ChungISIT03,Yamashitay-Luo,Luo-Pang}.
%In the presence of spectral mask constraints, the derivations in \cite{Yu,
%ChungISIT03, Yamashitay-Luo, Luo-Pang} cannot be applied and a solution for
%the system of nonlinear equations (\ref{sym_WF-sistem_mask}) may not exist for
%any set of channels and spectral masks. In general, not every strategic game
%admits a NE as defined in Definition \ref{NE def} \cite{Osborne, Aubin-book,
%Rosen}; however, for the game ${%
%%TCIMACRO{\TeXButton{G}{{\mathscr{G}}}}%
%%BeginExpansion
%{\mathscr{G}}%
%%EndExpansion
%}$ in (\ref{Game G}) the existence of pure-strategies NEs is guaranteed by the following.
Interestingly, the presence of spectral mask constraints does not
affect the existence of a pure-strategies NE of game $\mathscr{G}$,
as stated in the following. \begin{proposition} \label{Existence}The
game ${\mathscr{G}}$ in (\ref{Game G}) always admits at least one
pure-strategy NE, for any set of channel realizations, power
and spectral mask constraints. \end{proposition} \begin{proof} The
proof comes from standard results of game theory \cite{Osborne,Aubin-book}
and it is given in \cite{Scutari-Barbarossa-GT_PI}. %as shown next. For any given set of channel realizations, power and spectral mask
%constraints, the game ${%
%%TCIMACRO{\TeXButton{G}{{\mathscr{G}}}}%
%%BeginExpansion
%{\mathscr{G}}%
%%EndExpansion
%}$ in (\ref{Game G}) has the following properties: i) The payoff function of
%each player is continuous in the vector strategies of all the players, and it
%is (strictly) concave in its own strategy; ii) Each admissible strategy set ${%
%%TCIMACRO{\TeXButton{P}{{\mathscr{P}}}}%
%%BeginExpansion
%{\mathscr{P}}%
%%EndExpansion
%}_{q}$ in (\ref{admissible strategy set}) is a convex set. Hence, ${%
%%TCIMACRO{\TeXButton{G}{{\mathscr{G}}}}%
%%BeginExpansion
%{\mathscr{G}}%
%%EndExpansion
%}$ is a concave game, according to the definition in \cite{Rosen}, and this
%guarantees the existence of a pure-strategy NE.
\end{proof}

In general, the game ${\mathscr{G}}$ may admit multiple equilibria, depending
on the level of multiuser interference \cite{Scutari-Barbarossa-GT_PI}.
In the forthcoming sections, we provide sufficient conditions ensuring
the uniqueness of the NE and we address the problem of how to reach
such an equilibrium in a totally distributed way.

\section{Asynchronous Iterative Waterfilling}

\label{Sec:AIWFA} To reach the NE of game ${\mathscr{G}}$, we propose
a totally asynchronous distributed iterative waterfilling procedure,
which we name asynchronous Iterative WaterFilling Algorithm. The proposed
algorithm can be seen as an instance of the totally asynchronous scheme
of \cite{Bertsekas Book-Parallel-Comp}: all the users maximize their
own rate in a \emph{totally asynchronous} way. % via the single user waterfilling solution.  According to this
%asynchronous procedure,
More specifically, some users are allowed to update their strategy
more frequently than the others, and they might perform their updates
using \emph{outdated}  information about the interference caused from
the others. What we show is that the asynchronous IWFA converges to
a stable NE of game ${\mathscr{G}}$, \textit{whichever the updating
schedule is}, under rather mild conditions on the multiuser interference.
Interestingly, these conditions are also sufficient to guarantee the
\textit{uniqueness} of the NE.

To provide a formal description of the asynchronous IWFA, we need
to introduce some preliminary definitions. We assume, without any
loss of generality, that the set of times at which one or more users
update their strategies is the discrete set $\mathcal{T}=\mathbb{N}_{+}=\left\{ 0,1,2,\ldots\right\} .$
Let $\mathbf{p}_{q}^{(n)}$ denote the vector power allocation of
user $q$ at the $n$-th iteration, and let $\mathcal{T}_{q}\subseteq\mathcal{T}$
\ denote the set of times $n$ at which user $q$ updates his power
vector $\mathbf{p}_{q}^{(n)}$ (thus, implying that, at time $n\notin\mathcal{T}_{q},$
$\mathbf{p}_{q}^{(n)}$ is left unchanged). Let $\tau_{r}^{q}(n)$
denote the most recent time at which the interference from user $r$
is perceived by user $q$ at the $n$-th iteration (observe that $\tau_{r}^{q}(n)$
satisfies $0\leq\tau_{r}^{q}(n)\leq n$). Hence, if user $q$ updates
its power allocation at the $n$-th iteration, then it waterfills,
according to (\ref{WF_mask}), the interference level caused by the
power allocations of the others: \begin{equation}
\mathbf{p}_{-q}^{(\mathbf{\tau}^{q}(n))}\triangleq\left(\mathbf{p}_{1}^{(\tau_{1}^{q}(n))},\ldots,\mathbf{p}_{q-1}^{(\tau_{q-1}^{q}(n))},\mathbf{p}_{q+1}^{(\tau_{q+1}^{q}(n))},\ldots,\mathbf{p}_{Q}^{(\tau_{Q}^{q}(n))}\right).\label{p_q_interference}\end{equation}

The overall system is said to be totally asynchronous if the following
weak assumptions are satisfied for each $q$ \cite{Bertsekas Book-Parallel-Comp}:
%\begin{align}
%&  0\leq\tau_{r}^{q}(n)\leq n,\tag{A1}\label{A1}\\
%&  \lim_{k\rightarrow\infty}\tau_{r}^{q}(n_{k})=+\infty,\tag{A2}\label{A2}\\
%&  \left\vert T_{q}\right\vert =\infty. \tag{A3}\label{A3}%
%\end{align}
A1) $0\leq\tau_{r}^{q}(n)\leq n$; A2) $\lim_{k\rightarrow\infty}\tau_{r}^{q}(n_{k})=+\infty$;
and A3) $\left\vert \mathcal{T}_{q}\right\vert =\infty$; where $\{n_{k}\}$
is a sequence of elements in $\mathcal{T}_{q}$ that tends to infinity.
Assumptions A1$-$A3 are standard in asynchronous convergence theory
\cite{Bertsekas Book-Parallel-Comp}, and they are fulfilled in any
practical implementation. In fact, A1 simply indicates that, at any
given iteration $n$, each user $q$ can use only the interference
vectors $\mathbf{p}_{-q}^{(\mathbf{\tau}^{q}(n))}$ allocated by the
other users in the previous iterations (to preserve causality). Assumption
A2 states that, for any given iteration index $n_{k},$ the values
of the components of $\mathbf{p}_{-q}^{(\mathbf{\tau}^{q}(n))}$ in
(\ref{p_q_interference}) generated prior to $n_{k},$ are not used
in the updates of $\mathbf{p}_{q}^{(n)}$, when $n$ becomes sufficiently
larger than $n_{k}$; which guarantees that old information is eventually
purged from the system. Finally, assumption A3 indicates that no user
fails to update its own strategy as time $n$ goes on.

Given game ${{\mathscr{G}},}$ let $\mathcal{D}_{q}^{\min}\subseteq\{1,\cdots,N\}$
denote the set $\{1,\ldots,N\}$ (possibly) deprived of the carrier
indices that user $q$ would never use as the best response set to
\emph{any} strategies adopted by the other users \cite{Scutari-Barbarossa-GT_PI}:
\begin{equation}
\mathcal{D}_{q}^{\min}\triangleq\left\{ k\in\{1,\ldots,N\}:\exists\text{ }\mathbf{p}_{-q}\in{\mathscr{P}}_{-q}\text{ such that }\left[{\mathsf{WF}}_{q}\left(\mathbf{p}_{-q}\right)\right]_{k}\neq0\right\} ,\label{D_q}\end{equation}

\noindent with ${\mathsf{WF}}_{q}\left(\mathbf{\cdot}\right)$ defined in (\ref{WF_mask})
and ${\mathscr{P}}_{-q}\triangleq{\mathscr{P}}_{1}\times\cdots\times{\mathscr{P}}_{q-1}\times{\mathscr{P}}_{q+1}\times\cdots\times{\mathscr{P}}_{Q}$.
In \cite{Scutari-Barbarossa-GT_PI}, an iterative procedure
to obtain a set $\mathcal{D}_{q}$ such that $\mathcal{D}_{q}^{\min}\subseteq\mathcal{D}_{q}\subseteq\{1,\cdots,N\}$ is given.
Let the matrix $\mathbf{S}^{\max}\in\mathbb{R}^{Q\times Q}_{+}$ be
defined as \begin{equation}
\left[\mathbf{S}^{\max}\right]_{qr}\triangleq\left\{ \begin{array}{ll}
\max\limits _{k\in\mathcal{D}_{q}\cap\mathcal{D}_{r}}\dfrac{|\bar{H}_{qr}(k)|^{2}}{|\bar{H}_{qq}(k)|^{2}}\dfrac{d_{qq}^{\gamma}}{d_{qr}^{\gamma}}\dfrac{P_{r}}{P_{q}}, & \text{if }\ r\neq q,\\
0, & \text{otherwise,}\end{array}\right.\label{Hmax}\end{equation}
 with the convention that the maximum in (\ref{Hmax}) is zero if
$\ \mathcal{D}_{q}\cap\mathcal{D}_{r}$ is empty. In (\ref{Hmax}),
each set $\mathcal{D}_{q}$ can be chosen as any subset of $\{1,\cdots,N\}$
such that $\mathcal{D}_{q}^{\min}\subseteq\mathcal{D}_{q}\subseteq\{1,\cdots,N\}$,
with $\mathcal{D}_{q}^{\min}$ defined in (\ref{D_q}). Using the
above notation, the asynchronous IWFA is described in Algorithm 1
(where $\mathrm{N_{it}}$ denotes the number of iterations).

\begin{algo}{Asynchronous Iterative Waterfilling Algorithm} SSet
$n=0$ and $\mathbf{p}_{q}^{(0)}=$ any feasible power allocation;
\newline\texttt{for} $n=0:\mathrm{N_{it}},$ \newline \begin{equation}
\,\,\,\,\mathbf{p}_{q}^{(n+1)}=\left\{ \begin{array}{ll}
\mathsf{WF}_{q}\left(\mathbf{p}_{1}^{(\tau_{1}^{q}(n))},\ldots,\mathbf{p}_{q-1}^{(\tau_{q-1}^{q}(n))},\mathbf{p}_{q+1}^{(\tau_{q+1}^{q}(n))},\ldots,\mathbf{p}_{Q}^{(\tau_{Q}^{q}(n))}\right), & \text{if }n\in\mathcal{T}_{q},\\
\mathbf{p}_{q}^{(n)}, & \text{otherwise};\end{array}\right.\hspace{1cm}\forall q\in\Omega\label{AIWFA}\end{equation}
 \newline \texttt{end} \end{algo} 

The convergence of the algorithm is guaranteed under the following
sufficient conditions.

\begin{theorem} \label{Theo-AIWFA} Assume that the following condition
is satisfied: \begin{equation}
\rho\left(\mathbf{S}^{\max}\right)<1,\tag{C1}\label{(C0)}\end{equation}
 where $\mathbf{S}^{\max}$ is defined in (\ref{Hmax}) and $\rho\left(\mathbf{S}^{\max}\right)$
denotes the spectral radius%
\footnote{The spectral radius $\rho\left(\mathbf{S}\right)$ of the matrix $\mathbf{S}$
is defined as $\rho\left(\mathbf{S}\right)=\max\left\{ \left\vert \lambda\right\vert :\lambda\right.$
$\left.\in\mathrm{eig}\left(\mathbf{S}\right)\right\} $, with $\mathrm{eig}\left(\mathbf{S}\right)$
denoting the set of eigenvalues of $\mathbf{S}$ \cite{Horn-book}.%
} of $\mathbf{S}^{\max}$. Then, as $\mathrm{N_{it}}\rightarrow\infty$,
the asynchronous IWFA described in Algorithm 1 converges to the unique
NE of game $\mathscr{G}$, for any set of feasible initial conditions
and updating schedule. \end{theorem}

\begin{proof} The proof consists in showing that, under (\ref{(C0)}),
conditions of the Asynchronous Convergence Theorem in \cite{Bertsekas Book-Parallel-Comp}
are satisfied. A key point in the proof is given by the following
property of the multiuser waterfilling mapping $\mathsf{WF}\left(\mathbf{p}\right)=\left (\mathsf{WF}_{q}\left(\mathbf{p}_{-q}\right)\right) _{q\in\Omega}$
based on the interpretation of the waterfilling solution (\ref{WF_mask})
as a proper projector \cite{Scutari-Barbarossa-GT_PII}: \begin{equation}
\left\Vert \mathbf{\mathsf{WF}(p}^{(1)}\mathbf{)}-\mathbf{\mathsf{WF}(p}^{(2)}\mathbf{)}\right\Vert \leq\beta\left\Vert \mathbf{p}^{(1)}-\mathbf{p}^{(2)}\right\Vert ,\quad\forall\mathbf{p}^{(1)},\text{ }\mathbf{p}^{(2)}\in{\mathscr{P}}\text{ }\mathbf{,}\label{contraction_WF}\end{equation}
 where $\left\Vert \cdot\right\Vert $ is a proper vector norm and
$\beta$ is a positive constant,  which is less than $1$ if condition (\ref{(C0)}) is satisfied. See Appendix \ref{Proof_Theorems_1_and_2}
for the details. \end{proof}

\bigskip{}

To give additional insight into the physical interpretation of the
convergence conditions of Algorithm 1, we provide the following corollary
of Theorem \ref{Theo-AIWFA}.

\begin{corollary} \label{Corollary-C1-C2}A sufficient condition
for (\ref{(C0)}) in Theorem \ref{Theo-AIWFA} is given by one of the two following set of conditions:\begin{equation}
\hspace{-0.1cm}\dfrac{1}{w_{q}}\text{ }\!\!{\displaystyle \sum\limits _{r\neq q}}\max\limits _{k\in\mathcal{D}_{r}\cap\mathcal{D}_{q}}\left\{ \dfrac{|\bar{H}_{qr}(k)|^{2}}{|\bar{H}_{qq}(k)|^{2}}\right\} \dfrac{d_{qq}^{\gamma}}{d_{qr}^{\gamma}}\dfrac{P_{r}}{P_{q}}w_{r}<1,\quad\quad \forall q\in\Omega,\hspace{-0.15cm}\vspace{-0.2cm}\tag{C2}\label{cond_C1}\end{equation}
\begin{equation}
\hspace{-0.1cm}\dfrac{1}{w_{r}}\!\!\text{ }{\displaystyle \sum\limits _{q\neq r}}\max\limits _{k\in\mathcal{D}_{r}\cap\mathcal{D}_{q}}\left\{ \dfrac{|\bar{H}_{qr}(k)|^{2}}{|\bar{H}_{qq}(k)|^{2}}\right\} \dfrac{d_{qq}^{\gamma}}{d_{qr}^{\gamma}}\dfrac{P_{r}}{P_{q}}w_{q}<1,\quad\quad \,\forall r\in\Omega,\hspace{-0.15cm}\tag{C3}\label{cond_C2}\end{equation}
 where $\mathbf{w}\triangleq\lbrack w_{1},\ldots,w_{Q}]^{T}$ is any
positive vector.%
\footnote{The optimal positive vector $\mathbf{w}$ in (\ref{cond_C1})-(\ref{cond_C2})
is given by the solution of a geometric programming, as shown in \cite[Corollary 5]{Scutari-Barbarossa-GT_PII}%
} \end{corollary}

\noindent Note that, according to the definition of $\mathcal{D}_{q}$
in (\ref{Hmax}), one can always choose $\mathcal{D}_{q}=\{1,\ldots,N\}$
in (\ref{(C0)}) and (\ref{cond_C1})-(\ref{cond_C2}). However, less
stringent conditions are obtained by removing unnecessary carriers,
i.e., the carriers that, for the given power budget and interference
levels, are never going to be used. 

Recall that, if finite order constellations
are used, Theorem \ref{Theo-AIWFA} is still valid using the gap-approximation
method \cite{Forney,Goldsmith-Chua} as pointed out in Section \ref{Sec:Problem-Formulation}.
It is sufficient to replace each $\left\vert H_{qq}(k)\right\vert ^{2}$
in (\ref{(C0)}) with $\left\vert H_{qq}(k)\right\vert ^{2}/\Gamma_{q}.$

\medskip{}

\noindent \addtocounter{rem}{1}\textbf{Remark {\therem\ - Global
convergence and robustness of the algorithm}}: Even though the rate
maximization game in (\ref{Rate Game}) and the consequent waterfilling
mapping (\ref{WF_mask}) are nonlinear, condition (\ref{(C0)}) guarantees
the \emph{global} convergence of the asynchronous IWFA. Observe that
Algorithm 1 contains as special cases a plethora of algorithms, each
one obtained by a possible choice of the scheduling of the users in
the updating procedure (i.e., the parameters $\{\tau_{r}^{q}(n)\}$
and $\{\mathcal{T}_{q}\}$). The important result stated in Theorem
\ref{Theo-AIWFA} is that all the algorithms resulting as special
cases of the asynchronous IWFA are guaranteed to reach the unique
NE of the game, under the same set of convergence conditions (provided
that A1$-$A3 are satisfied), since condition (\ref{(C0)}) does
not depend on the particular choice of $\{\mathcal{T}_{q}\}$ and
$\{\tau_{r}^{q}(n)\}.$

\medskip{}

\addtocounter{rem}{1}
\noindent \textbf{Remark {\therem} - Physical
interpretation of convergence conditions: \ }As expected, the convergence
of the asynchronous IWFA and the uniqueness of NE are ensured if the
interferers are sufficiently far apart from the destinations. In fact,
from (\ref{cond_C1})-(\ref{cond_C2}) one infers that, for any given
set of channel realizations and power constraints, there exists a
distance beyond which the convergence of the asynchronous IWFA (and
the uniqueness of NE) is guaranteed, corresponding to the maximum
level of interference that may be tolerated by each receiver [as quantified, e.g., in (\ref{cond_C1})] or that may be generated by each transmitter [as quantified, e.g., in (\ref{cond_C2})]. But
the most interesting result coming from (\ref{(C0)}) and (\ref{cond_C1})-(\ref{cond_C2})
is that the convergence of the asynchronous IWFA is robust against
the worst normalized channels $|H_{qr}(k)|^{2}/$ $|H_{qq}(k)|^{2};$
in fact, the subchannels corresponding to the highest ratios $|H_{qr}(k)|^{2}/|H_{qq}(k)|^{2}$
(and, in particular, the subchannels where $|H_{qq}(k)|^{2}$ is vanishing)
do not necessarily affect the convergence of the algorithm, as their
carrier indices may not belong to the set $\mathcal{D}_{q}$.

\medskip{}

\addtocounter{rem}{1}
\noindent \textbf{Remark {\therem\ - Distributed
nature of the algorithm: }}Since the asynchronous IWFA is based on
the waterfilling solution (\ref{WF_mask}), it can be implemented
in a distributed way, where each user, to maximize his own rate, only
needs to measure the PSD of the overall interference-plus-noise and
waterfill over it. More interestingly, according to the asynchronous
scheme, the users may update their strategies using a potentially
outdated version of the interference. Furthermore, some users are
even allowed to update their power allocation more often than others,
without affecting the convergence of the algorithm. These features
strongly relax the constraints on the synchronization of the users'
updates with respect to those imposed, for example, by the simultaneous
or sequential updating schemes.

\noindent We can generalize the asynchronous IWFA given in Algorithm
1 by introducing a memory in the updating process, as given in Algorithm
2. We call this new algorithm \emph{smoothed} asynchronous IWFA.\bigskip{}

\begin{algo}{Smoothed Asynchronous Iterative Waterfilling Algorithm}
SSet $n=0$ and $\mathbf{p}_{q}^{(0)}=$ any feasible power allocation
and $\alpha_{q}\in[0,1),\,\forall q\in\Omega$; \newline\texttt{for}
$n=0:\mathrm{N_{it}},$ \newline \begin{equation}
\,\,\,\,\mathbf{p}_{q}^{(n+1)}=\left\{ \begin{array}{ll}
\alpha_{q}\mathbf{p}_{q}^{(n)}+(1-\alpha_{q})\mathsf{WF}_{q}\left(\mathbf{p}_{-q}^{(\mathbf{\tau}^{q}(n))}\right), & \text{if }n\in T_{q},\\
\mathbf{p}_{q}^{(n)}, & \text{otherwise},\end{array}\right.\hspace{1cm}\forall q\in\Omega;\label{AIWFA-memory}\end{equation}
 \newline \texttt{end} \end{algo}\bigskip{}

Each factor $\alpha_{q}\in\lbrack0,1)$ in Algorithm 2 can be interpreted
as a forgetting factor: The larger $\alpha_{q}$ is, the longer 
the memory of the algorithm is.%
\footnote{In this paper, we are only considering constant channels. Nevertheless,
in a time-varying scenario, (\ref{AIWFA-memory}) could be used to
smooth the fluctuations due to the channel variations. In such a case,
if the channel is fixed or highly stationary, it is convenient to
take $\alpha_{q}$ close to $1$; conversely, if the channel is rapidly
varying, it is better to take a small $\alpha_{q}$.%
} Interestingly the choice of $\alpha_{q}$'s does not affect the convergence
capability of the algorithm (although it may affect the speed of convergence \cite{Scutari-Barbarossa-GT_PII}), as proved in the following.

\begin{theorem} \label{Theo-AIWFA_memory} Assume that condition
of Theorem \ref{Theo-AIWFA} is satisfied. Then, as $\mathrm{N_{it}}\rightarrow\infty$,
the smoothed asynchronous IWFA described in Algorithm 2 converges to the unique
NE of game $\mathscr{G}$, for any set of feasible initial conditions,
updating schedule, and $\{\alpha_{q}\}_{q\in\Omega},$ with $\alpha_{q}\in\left[0,1\right),$
$\forall q\in\Omega$. \end{theorem}

\begin{proof} See Appendix \ref{Proof_Theorems_1_and_2}. \end{proof}

\medskip{}

\addtocounter{rem}{1}
\noindent \textbf{Remark {\therem\ - Asynchronous IWFA
in the presence of intercarrier interference:}} The proposed AIWFA
can be extended to the case where the transmission by the different
users contains time and frequency synchronization offsets. In \cite{Scutari-ITA07,Scutari-ICASSP07}
we showed that the Asynchronous IWFA is robust against the intercarrier
interference due to time and/or frequency offsets among the links
and we derived sufficient conditions guaranteeing its convergence
in the presence of such time/frequency misalignments.
\vspace{-0.5cm}

\section{Two Special Cases: Sequential and Simultaneous IWFAs}

\label{Sec:two examples} In this section, we specialize our asynchronous
IWFA to two special cases: the \emph{sequential} and the \emph{simultaneous}
IWFAs. As a by-product of the proposed unified framework, we show
that both algorithms converge to the NE under the same sufficient
conditions, that are larger than  conditions obtained for the convergence of the sequential IWFA in 
\cite{Yu,ChungISIT03}, \cite{Yamashitay-Luo}, \cite{Tse} (without considering the spectral mask constraints) and \cite{Luo-Pang}
(including the spectral mask constraints).

\medskip{}

\noindent \textbf{Sequential Iterative Waterfilling:} The sequential
IWFA is an instance of the Gauss-Seidel scheme by which each user
is sequentially updated \cite{Bertsekas Book-Parallel-Comp} based
on the waterfilling mapping (\ref{WF_mask}). In fact, in sequential
IWFA each player, sequentially and according to a fixed order, maximizes
his own rate (\ref{Rate}), performing the single-user waterfilling
solution in (\ref{WF_mask}), given the others as interference. This
scheme can also be seen as a particular case of the general asynchronous
IWFA with the following parameters: $\mathcal{T}_{q}=\left\{ kQ+q,\quad k\in\mathbb{N}_{+}\right\} =\left\{ q,Q+q,2Q+q,\ldots\right\} $
and $\tau_{r}^{q}(n)=n,$ $\forall r,q.$ Using this settings in Algorithm
1, the sequential IWFA can be written in compact form as in Algorithm
\ref{IWFA_Algo}.\bigskip{}

\begin{algo}{Sequential Iterative Waterfilling Algorithm}SSet $n=0$
and $\mathbf{p}_{q}^{(0)}=$ any feasible power allocation; \newline\texttt{for}
$n=0:\mathrm{N_{it}},$ \newline \begin{equation}
\,\,\,\,\mathbf{p}_{q}^{(n+1)}=\left\{ \begin{array}{ll}
\mathsf{WF}_{q}\left(\mathbf{p}_{-q}^{(n)}\right), & \text{if }(n+1)\,\text{mod}\, Q=q,\\
\mathbf{p}_{q}^{(n)}, & \text{otherwise},\end{array}\right.\hspace{1cm}\forall q\in\Omega;\end{equation}
 \newline \texttt{end}\label{IWFA_Algo} \end{algo}

\bigskip{}

\noindent \textbf{Simultaneous Iterative Waterfilling:} The simultaneous
IWFA can be interpreted as the synchronous Jacobi instance of the
asynchronous IWFA. In fact, in the simultaneous IWFA, all the users
update their own covariance matrix \emph{simultaneously} at each iteration,
performing the single user waterfilling solution (\ref{WF_mask}),
given the interference generated by the other users in the \emph{previous}
iteration. This is a particular case of the asynchronous IWFA in Algorithm
1 with the following parameters: $\mathcal{T}_{q}=\mathbb{N}_{+},$
\ and \ $\tau_{r}^{q}(n)=n,$ $\forall r,q,$ which leads to Algorithm
\ref{Algo_SIWFA}.\medskip{}

\begin{algo}{Simultaneous Iterative Waterfilling Algorithm} SSet
$n=0$ and $\mathbf{p}_{q}^{(0)}=$ any feasible power allocation;
\newline\texttt{for} $n=0:\mathrm{N_{it}},$ \newline\vspace{-0.2cm}
 \begin{equation}
\mathbf{p}_{q}^{(n+1)}={\mathsf{{WF}}}_{q}\left(\mathbf{p}_{-q}^{(n)}\right),\hspace{1cm}\forall q\in\Omega;\label{SIWFA}\end{equation}
 \texttt{end}\label{Algo_SIWFA} \end{algo} 

\medskip{}

By direct product of our unified framework, invoking Theorem \ref{Theo-AIWFA}
we obtain the following unified set of convergence conditions for
both sequential and simultaneous IWFAs \cite{Scutari-Barbarossa-GT_PII}.
\vspace{-0.1cm}

\begin{theorem} \label{Theo-Sim_seq_IWFA} Assume that condition
(\ref{(C0)}) of Theorem \ref{Theo-AIWFA} is satisfied. Then, as
$\mathrm{N_{it}}\rightarrow\infty$, the sequential and simultaneous IWFAs, described
in Algorithm \ref{IWFA_Algo} and \ref{Algo_SIWFA}, respectively,
converge geometrically to the unique NE of game ${\mathscr{G}}$ \ for
any set of feasible initial conditions and updating schedule. \end{theorem}

\medskip{}

\addtocounter{rem}{1}
\noindent \textbf{Remark {\therem\ - Algorithm robustness:
\ }}It follows form Theorem\textbf{\ }\ref{Theo-Sim_seq_IWFA}
that slight variations of the sequential or simultaneous IWFAs that
fall within the unified framework of the asynchronous IWFA, are still
guaranteed to converge, under the condition in Theorem \ref{Theo-AIWFA}.
This means that, using for example Algorithm \ref{IWFA_Algo}, if
the order in the users' updates changes during time, or some user
skips some update, or he uses an outdated version of the interference
PSD, this does not affect the convergence of the algorithm. What is
affected is only the convergence time. Moreover, as for the smoothed
asynchronous IWFA, also for the sequential and simultaneous IWFAs
described in Algorithm \ref{IWFA_Algo} and \ref{Algo_SIWFA}, respectively,
one can introduce a memory in the updating process \cite{Scutari-Barbarossa-GT_PII},
still guaranteeing convergence under conditions of Theorem\textbf{\
}\ref{Theo-Sim_seq_IWFA}.

\medskip{}

\addtocounter{rem}{1}
\noindent \textbf{Remark {\therem\ - }Comparison with
previous convergence conditions: \ }Algorithm \ref{IWFA_Algo} generalizes
the well-known sequential iterative waterfilling algorithm proposed
by Yu et al. in \cite{Yu} to the case where the spectral mask constraints
are explicitly taken into account. In fact, the algorithm in \cite{Yu}
can be obtained as a special case of Algorithm \ref{IWFA_Algo}, by
removing the spectral mask constraints in each set ${\mathscr{P}}_{q}$
in (\ref{admissible strategy set}), (i.e. setting $p_{q}^{\max}(k)=+\infty,$
$\forall k,q$), so that the waterfilling operator in (\ref{WF_mask})
becomes the classical waterfilling solution \cite{Cover}, i.e., $\begin{array}{c}
\mathsf{WF}_{q}\left(\mathbf{p}_{-q}\right)=\left(\mu_{q}\mathbf{1}_{N}-\boldsymbol{\mathsf{insr}}_{q}\right)^{+}\end{array},$ where $\left(x\right)^{+}=\max(0,x)$ and $\boldsymbol{\mathsf{insr}}_{q}\triangleq\lbrack\mathsf{insr}_{q}(1),\ldots,\mathsf{insr}_{q}(N)]^{T},$
with $\mathsf{insr}_{q}(k)=(\sigma_{w_{q}}^{2}(k)+\sum_{\, r\neq q}\left\vert H_{qr}(k)\right\vert ^{2}p_{r}(k))/\left\vert H_{qq}(k)\right\vert ^{2}$.
The convergence of the sequential IWFA has been studied in a number
of works, either in the absence \cite{Yu,ChungISIT03,Scutari-Barbarossa-SPAWC03,Yamashitay-Luo,Tse}
or in the presence \cite{Scutari-Barbarossa-GT_PII,Luo-Pang} of spectral
mask constraints. Interestingly, conditions in \cite{Yu,ChungISIT03,Scutari-Barbarossa-SPAWC03,Yamashitay-Luo,Tse}
and \cite{Luo-Pang} imply our condition (\ref{(C0)}), which is more relaxed as shown next.
Let \begin{equation}
\mathbf{\Upsilon}\triangleq\left(\mathbf{I}-\overline{\mathbf{S}}_{\text{low}}^{\max}\right)^{-1}\overline{\mathbf{S}}_{\text{upp}}^{\max},\label{Gamma_matrix}\end{equation}
 with $\overline{\mathbf{S}}_{\text{low}}^{\max}$ and $\overline{\mathbf{S}}_{\text{upp}}^{\max}$
denoting the strictly lower and strictly upper triangular part of
the matrix $\overline{\mathbf{S}}^{\max}$, respectively, and $\overline{\mathbf{S}}^{\max}$
is defined similar to $\mathbf{S}^{\max}$ in (\ref{Hmax}), but taking the maximum
over the whole set $\{1,\ldots,N\}.$ The relationship between
(sufficient) conditions for the convergence of sequential IWFA as
derived in \cite{Yu,ChungISIT03,Scutari-Barbarossa-SPAWC03,Yamashitay-Luo, Luo-Pang, Tse}
and condition (\ref{(C0)})%
\footnote{Recall that condition (\ref{(C0)}) guarantees also the convergence
of the more general asynchronous IWFA, as stated in Theorem \ref{Theo-AIWFA}.%
} is given in the following corollary of Theorem \ref{Theo-Sim_seq_IWFA}.

\noindent \begin{corollary} \label{Corollary:Conditions of the others}Sufficient
conditions for (\ref{(C0)}) in Theorem $1$ are \cite{Yu}$-$\cite{Scutari-Barbarossa-SPAWC03}%
\footnote{In \cite{Yu}, the authors derived conditions (\ref{Chung SF}) for
a game composed by $Q=2$ users and in the absence of spectral mask
constraints. %We have not considered conditions in \cite{Yamashitay-Luo},
%because they are stronger than (\ref{Chung SF}).%
}\begin{equation}
\max\limits _{k\in\{1,\ldots,N\}}\left\{ \dfrac{|\bar{H}_{rq}(k)|^{2}}{|\bar{H}_{qq}(k)|^{2}}\right\} \dfrac{d_{qq}^{\alpha}}{d_{rq}^{\alpha}}\dfrac{P_{r}}{P_{q}}<\dfrac{1}{Q-1},\hspace{1.5cm}\forall\text{ }r,q\in\Omega,\,q\neq r,\tag{C4}\label{Chung SF}\end{equation}
 or \cite{Yamashitay-Luo}%
\begin{equation}
\max\limits_{k\in\{1,\ldots,N\}}\left\{  \dfrac{|\bar{H}_{rq}(k)|^{2}}%
{|\bar{H}_{qq}(k)|^{2}}\right\}  \dfrac{d_{qq}^{\alpha}}{d_{rq}^{\alpha}%
}\dfrac{P_{r}}{P_{q}}<\dfrac{1}{2Q-3},\hspace{1.5cm}\forall\text{ }r,q\in\Omega,\,q\neq r, \tag{C5}\label{Luo-Yamashitay}%
\end{equation}
or \cite{Luo-Pang}\begin{equation}
\rho\left(\mathbf{\Upsilon}\right)<1,\tag{C6}\label{Luo_Cond}\end{equation}
 where $\mathbf{\Upsilon}$ is defined in (\ref{Gamma_matrix}). \end{corollary}

\begin{proof} See Appendix \ref{proof_Corollary:Conditions of the others}.
\end{proof}

\medskip{}

Since the convergence conditions in Corollary \ref{Corollary:Conditions of the others}
depend on the channel realizations $\left\{ \bar{H}_{qr}(k)\right\} $
and on the distances $\left\{ d_{qr}\right\} ,$ there is a nonzero
probability that they are not satisfied for a given channel realization,
drawn from a given probability space. To compare the range of validity
of our conditions vs. the conditions available in the literature,
we tested them over a set of channel impulse responses generated as
vectors composed of $L=6$ i.i.d. complex Gaussian random variables
with zero mean and unit variance (multipath Rayleigh fading model).
Each user transmits over a set of $N=16$ subcarriers. We consider
a multicell cellular network as depicted in Figure \ref{Fig_1a},
composed by $7$ (regular) hexagonal cells, sharing the same band.
Hence, the transmissions from different cells typically interfere
with each other. For the simplicity of representation, we assume that
in each cell there is only one active link, corresponding to the transmission
from the base station (BS) to a mobile terminal (MT) placed at a corner
of the cell. According to this geometry, each MT receives a useful
signal that is comparable, in average sense, with the interference
signal transmitted by the BSs of two adjacent cells. The overall network
is modeled as a set of seven wideband interference channels. In Figure
\ref{Fig_1b}, we plot the probability that conditions (\ref{(C0)}),
(\ref{Chung SF}) and (\ref{Luo_Cond}) are satisfied versus the (normalized)
distance $r\in[0,1)$ (see Figure \ref{Fig_1a}), between each MT
and his BS (assumed to be equal for all the MT/BS pairs). We tested
our conditions considering the set $\mathcal{D}_{q},$ obtained using
the algorithm described in \cite{Scutari-Barbarossa-GT_PI}. %
\begin{figure}
 \vspace{-0.8cm}

\begin{centering}
\subfigure[Multicell cellular system ]{\includegraphics[width=7cm,height=7cm]{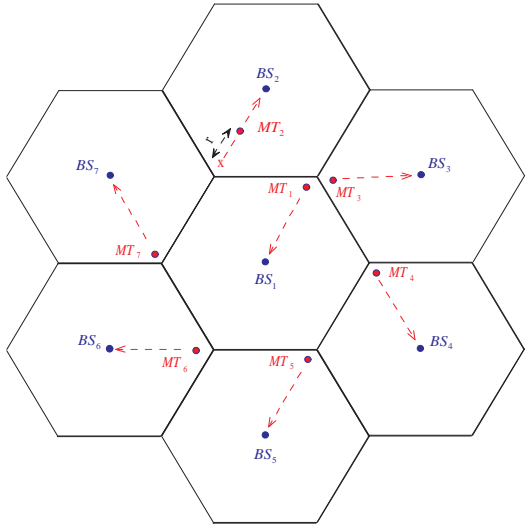}\label{Fig_1a}}
\hspace{0.3cm}\subfigure[Probability of (\protect\ref{(C0)}), (\protect\ref{Chung SF}) and (\protect\ref{Luo_Cond}) versus $r$.]{
\includegraphics[width=9.2cm,height=7cm]{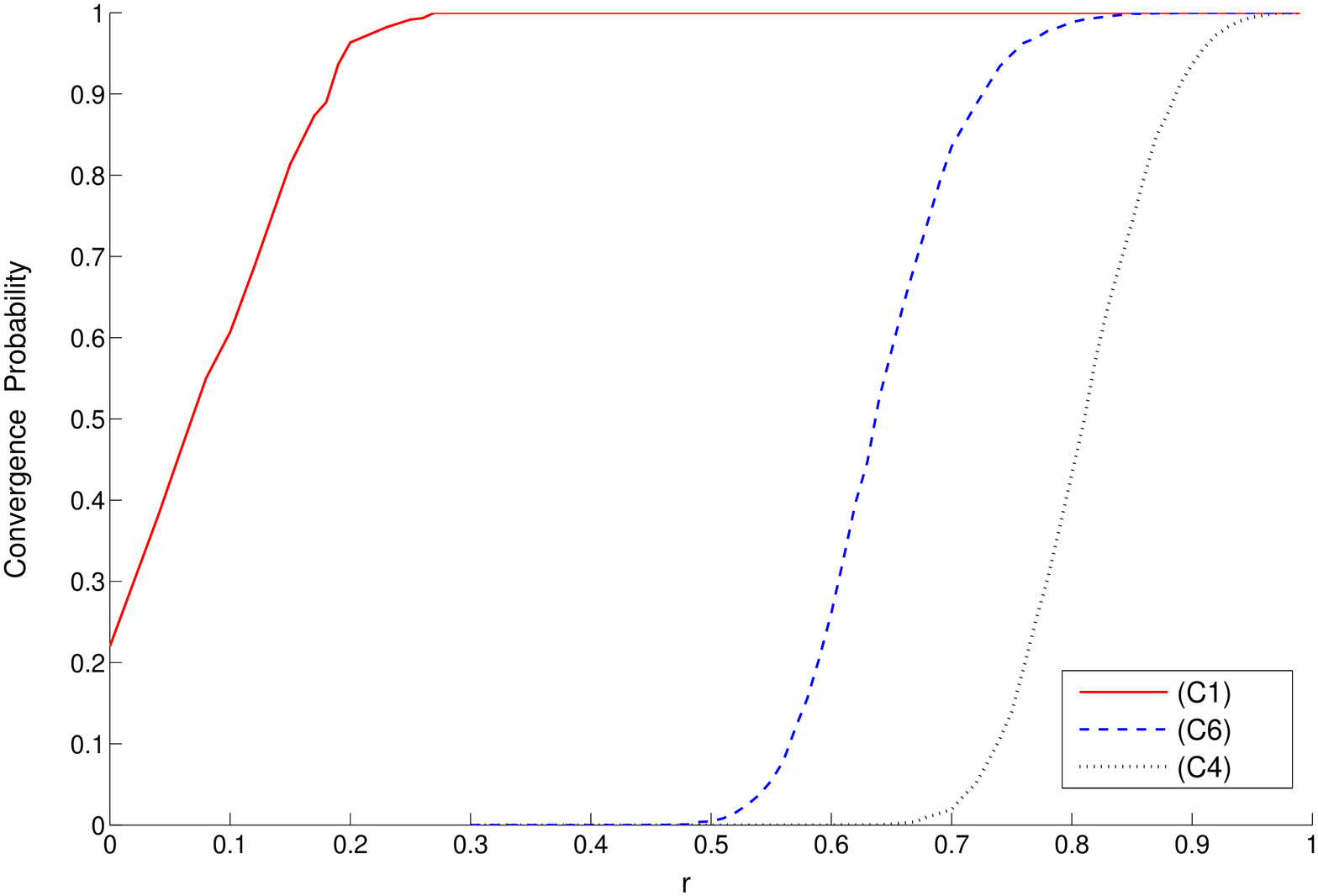}\label{Fig_1b}}\vspace{-0.2cm}
 
\par\end{centering}

\caption{{\protect{\small Probability of (\protect\ref{(C0)}), (\protect\ref{Chung SF})
and (\protect\ref{Luo_Cond}) versus $r$ {[}subplot (b)] for a $7$
cell (downlink) cellular system {[}subplot (a)]; $Q=7$ , $N=16,$
$\gamma=2.5,$ $P_{q}=P_{r}$, $\Gamma_{q}=1$, $P_{q}/\sigma_{q}^{2}=7$dB,
$\forall r,q\in\Omega,$ and $\boldsymbol{w}=1$.}}}

%\label{Figure_1}

\end{figure}

%\begin{figure}[ptb]
%\vspace{-0.5cm}
%\par
%\begin{center}
%\includegraphics[trim=0.000000in 0.000000in 0.000000in
%-0.212435in,height=9cm, width=11cm]{Figure_2b.eps}
%\end{center}
%\par
%\vspace{-0.8cm}\caption{{\protect\footnotesize Probability of (C1), (C4), and
%(C6) versus ... new simulaions to be done... }}%
%\label{Figure_1}%
%\end{figure}
As expected, the probability of guaranteeing convergence increases
as each MT approaches his BS (i.e., $r\rightarrow1$). What is worthwhile
noticing is that our condition (\ref{(C0)}) significantly enlarges
(\ref{Chung SF}) and (\ref{Luo_Cond}), since the probability that
(\ref{(C0)}) is satisfied is always much larger than (\ref{Chung SF})
and (\ref{Luo_Cond}).

\medskip{}

\addtocounter{rem}{1}
\noindent \textbf{Remark {\therem\ - }Sequential versus
simultaneous IWFA:} Since the simultaneous IWFA in Algorithm \ref{Algo_SIWFA}
is still based on the waterfilling solution (\ref{WF_mask}), it keeps
the most appealing features of the sequential IWFA, namely its low-complexity
and distributed nature. In addition, it allows all the users to update
their optimal power allocation \textit{simultaneously}. Hence, the
simultaneous IWFA is faster than the sequential IWFA, especially if
the number of active users in the network is large. A quantitative
comparison between the sequential and simultaneous IWFAs, in terms
of convergence speed, is given in \cite{Scutari-Barbarossa-GT_PII}.
In \cite{Scutari-ITA07}, we also provided a closed form expression
of the error estimates as a function of the iteration index, obtained
by the sequential and simultaneous IWFAs.

\section{Conclusion}

\noindent \label{Sec:Conclusion}In this paper, we have studied the
maximization of the information rates for the Gaussian frequency-selective
interference channel, given constraints on the transmit power and
the spectral masks on each link. We have formulated the optimization
problem as a strategic noncooperative game and we have proposed a
novel, totally asynchronous iterative distributed algorithm, named
asynchronous IWFA, to reach the Nash equilibria of the game. This
algorithm contains as special cases the well-known sequential IWFA
and the recently proposed simultaneous IWFA, where the users update
their strategies sequentially and simultaneously, respectively. The
main advantage of the proposed algorithm is that no rigid scheduling
in the updates of the users is required: Users are allowed to choose
their own strategies whenever they want and some users may even use
an outdated version of the measured multiuser interference. This relaxes
the coordination requirements among the users significantly. Finally,
we have provided the conditions ensuring the global convergence of
the asynchronous IWFA to the unique NE of the game. Interestingly,
our convergence conditions do not depend on the specific updating
scheduling performed by the users and, hence, they represent a unified
set of convergence conditions for all the algorithms that can be seen
as special cases of the asynchronous IWFA.\bigskip

\noindent \textbf{\Large Appendix} 

\appendix
%dummy comment inserted by tex2lyx to ensure that this paragraph is not empty
\vspace{-0.4cm}

\section{Proof of Theorems \ref{Theo-AIWFA} and \ref{Theo-AIWFA_memory}}

\label{Proof_Theorems_1_and_2}We start with some definitions and
intermediate results that will be instrumental to prove Theorems \ref{Theo-AIWFA}
and \ref{Theo-AIWFA_memory}.

\noindent \textbf{Properties of the waterfilling mapping.} For technical
reasons, we define the admissible set ${\mathscr{P}}^{\text{eff}}$
$={\mathscr{P}}_{1}^{\text{eff}}\times\cdots\times{\mathscr{P}}_{Q}^{\text{eff}}\subseteq{\mathscr{P},}$
where \begin{equation}
{\mathscr{P}}_{q}^{\text{eff}}\triangleq\left\{ \mathbf{p}_{q}\in{\mathscr{P}}_{q}\text{ with }p_{q}(k)=0\text{ }\forall k\notin\mathcal{D}_{q}\right\} ,\label{X_q set}\end{equation}
 is the subset of ${\mathscr{P}}_{q}$ containing all the feasible
power allocations of user $q$, with zero power over the carriers
that user $q$ would never use in any of its waterfilling solutions (\ref{WF_mask}), against
any admissible strategy of the others. Observe that the game does
not change if we use ${\mathscr{P}}^{\text{eff}}$ instead of the
original ${\mathscr{P}}$. For any given $\{\alpha_{q}\}_{q\in\Omega},$
with $\alpha_{q}\in\lbrack0,1),$ let $\mathbf{T}(\mathbf{p})=\left(\mathbf{T}_{q}(\mathbf{p})\right)_{q\in\Omega}:{\mathscr{P}}^{\text{eff}}{\mapsto\mathscr{P}}^{\text{eff}}$
be the mapping defined, for each $q,$ as\begin{equation}
\mathbf{T}_{q}(\mathbf{p})\triangleq\alpha_{q}\ \mathbf{p}_{q}+(1-\alpha_{q})\mathsf{WF}_{q}\left(\mathbf{p}_{-q}\right),\quad\mathbf{p}\in{\mathscr{P}}^{\text{eff}}\text{,}\label{Mapping_T}\end{equation}
 where $\mathsf{WF}_{q}\left(\mathbf{p}_{-q}\right):{\mathscr{P}}_{-q}^{\text{eff}}{\mapsto\mathscr{P}}_{q}^{\text{eff}}$
\ is the waterfilling operator defined in (\ref{WF_mask}). Observe
that all the Nash equilibria of game ${\mathscr{G}}$ correspond to
the fixed points in ${\mathscr{P}}^{\text{eff}}$ of the mapping\ $\mathbf{T}$
in (\ref{Mapping_T}). Hence, the existence of at least one fixed
point for $\mathbf{T}$ is guaranteed by {Proposition \ref{Existence}}.

Given $\mathbf{T}$ in (\ref{Mapping_T}) and some $\mathbf{w}\triangleq[w_{q},\ldots,w_{Q}]^{T}>\mathbf{0}$,
let $\left\Vert \cdot\right\Vert _{2,\text{block}}^{\mathbf{w}}$
denote the (vector) block-maximum norm$,$ defined as \cite{Bertsekas Book-Parallel-Comp}
\begin{equation}
\left\Vert \mathbf{T}(\mathbf{p})\right\Vert _{2,\text{block}}^{\mathbf{w}}\triangleq\max_{q\in\Omega}\frac{\left\Vert \mathbf{T}_{_{q}}(\mathbf{p})\right\Vert _{2}}{w_{q}},\text{ }\quad\label{block_max_weight_norm}\end{equation}
 where $\left\Vert \mathbf{\cdot}\right\Vert _{2}$ is the Euclidean
norm. Let $\left\Vert \mathbf{\cdot}\right\Vert _{\infty,\text{vec}}^{\mathbf{w}}$
be the \emph{vector} weighted maximum norm$,$ defined as \cite{Horn-book}
\begin{equation}
\left\Vert \mathbf{x}\right\Vert _{\infty,\text{vec}}^{\mathbf{w}}\triangleq\max_{q\in\Omega}\frac{\left\vert x_{q}\right\vert }{w_{q}},\quad\mathbf{w>0,}\text{ }\quad\mathbf{x\in\mathbb{R}}^{Q},\label{weighted_infinity_vector_norm}\end{equation}
 and let $\left\Vert \mathbf{\cdot}\right\Vert _{\infty,\text{mat}}^{\mathbf{w}}$
denote the \emph{matrix} norm induced by $\left\Vert \cdot\right\Vert _{\infty,\text{vec}}^{\mathbf{w}},$
defined as \cite{Horn-book} \begin{equation}
\left\Vert \mathbf{A}\right\Vert _{\infty,\text{mat}}^{\mathbf{w}}\triangleq\max_{q}\frac{1}{w_{q}}{\displaystyle \sum\limits _{r=1}^{Q}}[\mathbf{A}]_{qr}\text{ }w_{r},\text{ }\quad\mathbf{A\in\mathbb{R}}^{Q\times Q}.\label{H_max_weight_norm}\end{equation}
 Finally, we introduce the matrix $\mathbf{S}_{\mathbf{\alpha}}^{\max},$
defined as \begin{equation}
\mathbf{S}_{\mathbf{\alpha}}^{\max}\triangleq\mathbf{D}_{\mathbf{\alpha}}+(\mathbf{I-D}_{\mathbf{\alpha}})\mathbf{S}^{\max},\text{ \ with \ }\mathbf{D}_{\mathbf{\alpha}}\triangleq\operatorname*{diag}(\alpha_{q}\ldots\alpha_{Q}),\label{Hmax_alpha}\end{equation}
 where $\mathbf{S}^{\max}$ is defined in (\ref{Hmax}). The block-contraction
property of mapping $\mathbf{T}$ in (\ref{Mapping_T}) is given in
the following theorem that comes directly from \cite[Proposition 2]{Scutari-Barbarossa-GT_PII}.

\begin{theorem} [Contraction property of mapping $\mathbf{T}$]\label{Theorem_contraction}Given
$\mathbf{w}\triangleq[w_{1},\ldots,w_{Q}]^{T}\mathbf{>0}$ and $\{\alpha_{q}\}_{q\in\Omega},$
with $\alpha_{q}\in\lbrack0,1),$ the mapping $\mathbf{T}$ defined
in (\ref{Mapping_T}) satisfies\begin{equation}
\left\Vert \mathbf{T(p}^{(1)}\mathbf{)}-\mathbf{T(p}^{(2)}\mathbf{)}\right\Vert _{2,\mathrm{block}}^{\mathbf{w}}\leq\Vert\mathbf{S}_{\mathbf{\alpha}}^{\max}\Vert_{\infty,\text{\emph{mat}}}^{\mathbf{w}}\left\Vert \mathbf{p}^{(1)}-\mathbf{p}^{(2)}\right\Vert _{2,\mathrm{block}}^{\mathbf{w}},\quad\forall\mathbf{p}^{(1)},\text{ }\mathbf{p}^{(2)}\in{\mathscr{P}}^{\text{eff}}\text{ }\mathbf{,}\label{b_modulus_contraction}\end{equation}
 where $\left\Vert \cdot\right\Vert _{2,\text{block}}^{\mathbf{w}},$
$\left\Vert \cdot\right\Vert _{\infty,\text{mat}}^{\mathbf{w}}$ and
$\mathbf{S}_{\mathbf{\alpha}}^{\max}$ are defined in (\ref{block_max_weight_norm}),
(\ref{H_max_weight_norm}) and (\ref{Hmax_alpha}), respectively.
If $\Vert\mathbf{S}_{\mathbf{\alpha}}^{\max}\Vert_{\infty,\text{\emph{mat}}}^{\mathbf{w}}<1\mathbf{,}$
then mapping $\mathbf{T\ }$is a \textit{block-contraction with modulus}
$\Vert\mathbf{S}_{\mathbf{\alpha}}^{\max}\Vert_{\infty,\text{\emph{mat}}}^{\mathbf{w}}.$
\end{theorem}

\bigskip{}

\noindent \textbf{Asynchronous convergence theorem} \cite{Bertsekas Book-Parallel-Comp}.
Let $\mathcal{X}_{1},$ $\mathcal{X}_{2},\ldots,\mathcal{X}_{Q}$
be given sets, and let $\mathcal{X}$ be their Cartesian product,
i.e., \begin{equation}
\mathcal{X=X}_{1}\times\mathcal{X}_{2}\times\ldots\times\mathcal{X}_{Q}.\label{X_set}\end{equation}
 Let $\mathbf{f}_{q}:\mathcal{X\mapsto X}_{q}$ be a given vector
function and let $\mathbf{f}:\mathcal{X\mapsto X}$\ be the mapping
from $\mathcal{X}$\ to $\mathcal{X}$, defined as $\mathbf{f(x)=}\left(\mathbf{f}_{1}\mathbf{\mathbf{(x),\ldots,}f}_{Q}\mathbf{\mathbf{(x)}}\right),$
and assumed to admit a fixed point $\mathbf{x}^{\star}=\mathbf{f}(\mathbf{x}^{\star}).$
Consider the following distributed asynchronous iterative algorithm
to reach $\mathbf{x}^{\star}$\vspace{-0.2cm}
 \begin{equation}
\hspace{-0.2cm}\mathbf{x}_{q}^{(n+1)}=\left\{ \hspace{-0.05cm}\begin{array}{ll}
\mathbf{f}_{q}\left(\mathbf{x}_{1}^{(\tau_{1}^{q}(n))},\ldots,\mathbf{x}_{Q}^{(\tau_{Q}^{q}(n))}\right), & \text{if }n\in\mathcal{T}_{q},\\
\,\mathbf{x}_{q}^{(n)}, & \text{otherwise,}\end{array}\right.,\forall q\in\Omega;\label{Theorem_async_conv}\end{equation}
 with $0\leq\tau_{r}^{q}(n)\leq n$ and $\mathcal{T}_{q}$ denoting
the set of times $n$ at which $\mathbf{x}_{q}^{(n)}$ is updated
and satisfying A1$-$A3 of Section \ref{Sec:AIWFA}. Assume that:

\begin{description}
\item [{C.1}] (\emph{Nesting Condition}) There exists a sequence of nonempty
sets $\left\{ \mathcal{X}(n)\right\} $ with \begin{equation}
\ldots\subset\mathcal{X}(n+1)\subset\mathcal{X}(n)\subset\ldots\subset\mathcal{X},\label{contraction sets}\end{equation}
 satisfying the next two conditions.
\item [{C.2}] $($\emph{Synchronous Convergence Condition}$)$ \begin{equation}
\mathbf{f(x)\in}\mathcal{X}(n+1),\qquad\forall n,\text{ and }\mathbf{x\in}\mathcal{X}(n).\label{Synchronous Convergence Condition}\end{equation}
 Furthermore, if $\{\mathbf{y}^{(n)}\}$ is a sequence such that $\mathbf{y}^{(n)}\mathbf{\in}\mathcal{X}(n),$
for every $n,$ then every limit point of $\{\mathbf{y}^{(n)}\}$
is a fixed point of $\mathbf{f(\cdot).}$
\item [{C.3}] $($\emph{Box Condition}$)$ For every $n$ there exist sets
$\mathcal{X}_{q}(n)\subset\mathcal{X}_{q}$ such that\begin{equation}
\mathcal{X}(n)=\mathcal{X}_{1}(n)\times\ldots\times\mathcal{X}_{Q}(n).\label{Box conditions}\end{equation}

\end{description}
Then we have the following\vspace{-0.1cm}
.

\begin{theorem} [{\cite[Proposition 2.1]{Bertsekas Book-Parallel-Comp}}]\label{Theo-Bertsekas}
If the Synchronous Convergence Condition (\ref{Synchronous Convergence Condition})
and the Box Condition (\ref{Box conditions}) are satisfied, and the
starting point $\mathbf{x}^{(0)}\triangleq\left(\mathbf{x}_{1}^{(0)},\ldots,\mathbf{x}_{Q}^{(0)}\right)$
of the algorithm (\ref{Theorem_async_conv}) belongs to $\mathcal{X}(0)$,
then every limit point of $\{\mathbf{x}^{(n)}\}$ given by (\ref{Theorem_async_conv})
is a fixed point of $\mathbf{f(}\cdot\mathbf{).}$ \end{theorem}

\bigskip{}

\vspace{-0.2cm}
We are now ready to prove Theorems 1 and 2 through the following two
steps:

\noindent \emph{Step 1.} We first show that the asynchronous IWFA
in Algorithms 1 and 2 is an instance of the totally asynchronous iterative
algorithm in (\ref{Theorem_async_conv}). Then, using Theorem \ref{Theorem_contraction},
we derive sufficient conditions for \textbf{C.1}-\textbf{C.3}.

\noindent \emph{Step 2.} Invoking Theorem \ref{Theo-Bertsekas},\ we
complete the proof showing that the asynchronous IWFA converges to
the unique NE of ${\mathscr{G}}$ from any starting point, provided
that condition (\ref{(C0)}) is satisfied.

\bigskip{}

\noindent \emph{Step 1.} It is straightforward to see that the asynchronous
IWFA coincides with the algorithm given in (\ref{Theorem_async_conv}),
under the following identifications\vspace{-0.3cm}
 \begin{align}
 & \begin{array}{l}
\mathbf{x}_{q}\Leftrightarrow\mathbf{p}_{q},\quad\mathbf{x}_{q}^{\star}\Leftrightarrow\mathbf{p}_{q}^{\star},\quad\mathcal{X}_{q}\Leftrightarrow{\mathscr{P}}_{q}^{\text{eff}},\medskip\\
\mathbf{f}_{q}\mathbf{(x)}\Leftrightarrow\mathbf{T}_{q}(\mathbf{p}),\end{array}\quad\medskip\forall q\in\Omega,\bigskip\label{equivalences}\\
\  & \text{\ }\mathcal{X}\Leftrightarrow{\mathscr{P}}^{\text{eff}}={\mathscr{P}}_{1}^{\text{eff}}\times\ldots\times{\mathscr{P}}_{Q}^{\text{eff}},\nonumber \end{align}
 where ${\mathscr{P}}_{q}^{\text{eff}}$ and $\mathbf{T}_{q}(\mathbf{p})$
are defined in (\ref{X_q set}) and (\ref{Mapping_T}), respectively.
Observe that, to study the convergence of the asynchronous IWFA, there
is no loss of generality in considering the mapping $\mathbf{T}$
defined in ${\mathscr{P}}^{\text{eff}}\subset{\mathscr{P}}$ instead
of ${\mathscr{P},}$ since all the points produced by the algorithm
(except possibly the initial point, which does not affect the convergence
of the algorithm in the subsequent iterations) as well as the Nash
equilibria of the game are confined, by definition, in ${\mathscr{P}}^{\text{eff}}$.

We consider now conditions \textbf{C.1}-\textbf{C.3 \ }separately.

\noindent \textbf{C.1} (\emph{Nested Condition}) Let $\mathbf{p}^{\star}=\left(\mathbf{p}_{1}^{\star},\ldots,\mathbf{p}_{Q}^{\star}\right)\in{\mathscr{P}}^{\text{eff}}$
be a fixed point of $\mathbf{T}$ in (\ref{Mapping_T}) (or, equivalently
of $\mathbf{f}_{q}$ in (\ref{Theorem_async_conv})) and $\mathbf{p}^{(0)}=\left(\mathbf{p}_{1}^{(0)},\ldots,\mathbf{p}_{Q}^{(0)}\right)\in{\mathscr{P}}^{\text{eff}}$
be any starting point of the asynchronous IWFA. Using the block-maximum
norm $\left\Vert \mathbf{\cdot}\right\Vert _{2,\text{block}}^{\mathbf{w}}$
as defined in (\ref{block_max_weight_norm}), where $\mathbf{w=[}w_{1},\ldots,w_{q}\mathbf{]}^{T}$
is any positive vector, we define the set $\mathcal{X}(n)$ in (\ref{contraction sets})
as \begin{equation}
\mathcal{X}(n)=\left\{ \mathbf{p}\in{\mathscr{P}}^{\text{eff}}:\left\Vert \mathbf{p}-\mathbf{p}^{\star}\right\Vert _{2,\text{block}}^{\mathbf{w}}\leq\beta^{n}\Vert\mathbf{p}^{(0)}-\mathbf{p}^{\star}\Vert_{2,\text{block}}^{\mathbf{w}}\right\} \subset{\mathscr{P}}^{\text{eff}},\label{our-set-X}\end{equation}
 with\begin{equation}
\beta=\beta(\mathbf{w,S}_{\mathbf{\alpha}}^{\max})\triangleq\left\Vert \mathbf{S}_{\mathbf{\alpha}}^{\max}\right\Vert _{\infty}^{\mathbf{w}},\label{alpha_w,H}\end{equation}
 and $\mathbf{S}_{\mathbf{\alpha}}^{\max}$ defined in (\ref{Hmax_alpha}).
It follows from (\ref{our-set-X})  that if \begin{equation}
\beta^{n+1}\Vert\mathbf{p}^{(0)}-\mathbf{p}^{\star}\Vert_{2,\text{block}}^{\mathbf{w}}<\beta^{n}\Vert\mathbf{p}^{(0)}-\mathbf{p}^{\star}\Vert_{2,\text{block}}^{\mathbf{w}},\quad\forall n=0,1,...,\label{inequality_set}\end{equation}
 then we obtain the desired result, i.e., \[
\mathcal{X}(n+1)\subset\mathcal{X}(n)\subset{\mathscr{P}}^{\text{eff}},\quad\forall n=0,1,....\]
 A necessary and sufficient condition for (\ref{inequality_set})
is\begin{equation}
\beta<1.\label{SF_alpha}\end{equation}
 We will assume in the following that (\ref{SF_alpha}) is satisfied.

\noindent \textbf{C.2} (\emph{Synchronous Convergence Condition})
 Let $\mathbf{p}^{(n)}\in\mathcal{X}(n).$ Then, from (\ref{our-set-X}),
it must be that\begin{equation}
\left\Vert \mathbf{p}^{(n)}-\mathbf{p}^{\star}\right\Vert _{2,\text{block}}^{\mathbf{w}}\leq\beta^{n}\Vert\mathbf{p}^{(0)}-\mathbf{p}^{\star}\Vert_{2,\text{block}}^{\mathbf{w}}.\label{e_n_inf}\end{equation}
 Let $\mathbf{p}^{(n+1)}=\mathbf{T(\mathbf{p}}^{(n)}\mathbf{)}$.
Then, we have\begin{equation}
\left\Vert \mathbf{p}^{(n+1)}-\mathbf{p}^{\star}\right\Vert _{2,\text{block}}^{\mathbf{w}}=\left\Vert \mathbf{T}(\mathbf{p}^{(n)})-\mathbf{p}^{\star}\right\Vert _{2,\text{block}}^{\mathbf{w}}\leq\beta\left\Vert \mathbf{p}^{(n)}-\mathbf{p}^{\star}\right\Vert _{2,\text{block}}^{\mathbf{w}}\leq\beta^{n+1}\Vert\mathbf{p}^{(0)}-\mathbf{p}^{\star}\Vert_{2,\text{block}}^{\mathbf{w}},\label{ineq_e3}\end{equation}
 where the first and the second inequalities follow from Theorem \ref{Theorem_contraction}
(using (\ref{b_modulus_contraction}) with $\mathbf{p}^{\star}=\mathbf{T}\left(\mathbf{p}^{\star}\right)$)
and (\ref{e_n_inf}), respectively. Hence, $\mathbf{p}^{(n+1)}\in\mathcal{X}(n+1),$
as required in (\ref{Synchronous Convergence Condition}). Furthermore,
since \[
\lim_{n\rightarrow\infty}\left\Vert \mathbf{p}^{(n)}-\mathbf{p}^{\star}\right\Vert _{2,\text{block}}^{\mathbf{w}}=0,\text{ with }\mathbf{p}^{(n)}\in\mathcal{X}(n),\text{ }\forall n,\]
 the sequence $\{\mathbf{p}^{(n)}\}$ generated from $\mathbf{p}^{(0)}$
by the mapping $\mathbf{T}$ using the simultaneous updating scheme
in (\ref{Synchronous Convergence Condition}) must converge to $\mathbf{p}^{\star}$.
Moreover, it follows from (\ref{SF_alpha}) and Theorem  \ref{Theorem_contraction} that the fixed point $\mathbf{p}^{\star}$
of $\mathbf{T}$ is unique (implied from the fact that the mapping $\mathbf{T}$ is a  block-contraction \cite[Proposition 1.1]{Bertsekas Book-Parallel-Comp}).

\vspace{0.25cm}
\noindent \textbf{C.3} (\emph{Box Condition}) For every $n,$ the set $\mathcal{X}(n)$
in (\ref{our-set-X}) can be decomposed as $\mathcal{X}(n)=\mathcal{X}_{1}(n)\times\ldots\times\mathcal{X}_{Q}(n),$
with
\noindent \begin{equation}
\mathcal{X}_{q}(n)=\left\{ \mathbf{p}_{q}\in{\mathscr{P}}_{q}^{\text{eff}}:\frac{\left\Vert \mathbf{p}_{q}-\mathbf{p}_{q}^{\star}\right\Vert _{2}}{w_{q}}\leq\beta^{n}\Vert\mathbf{p}^{(0)}-\mathbf{p}^{\star}\Vert_{2,\text{block}}^{\mathbf{w}}\right\} \subset{\mathscr{P}}_{q}^{\text{eff}},\quad\forall q\in\Omega.\end{equation}

\noindent \emph{Step 2.} Under (\ref{SF_alpha}), Theorem \ref{Theo-Bertsekas}
is satisfied if the starting point $\mathbf{p}^{(0)}$ of the asynchronous
IWFA is such that $\mathbf{p}^{(0)}\in{\mathscr{P}}^{\text{eff}}.$
The asynchronous IWFA, as given in Algorithm 1 and Algorithm 2, \ is
allowed to start from any arbitrary point $\mathbf{p}^{(0)}$ in ${\mathscr{P}}.$
However, after the first iteration from all the users, the asynchronous
IWFA provides a point in ${\mathscr{P}}^{\text{eff}},$ for any $\mathbf{p}^{(0)}\in{\mathscr{P}}.$
Hence, under (\ref{SF_alpha}), the asynchronous IWFA satisfies Theorem
\ref{Theo-Bertsekas}, after the first iteration, which still guarantees
that every limit point of the sequence generated by the asynchronous
IWFA is a NE of the game ${\mathscr{G}}$. Since condition (\ref{SF_alpha})
is also sufficient for the uniqueness of the NE [recall that, under (\ref{SF_alpha}),  the mapping $\mathbf{T}$ in (\ref{Mapping_T}) is a  block-contraction], the asynchronous
IWFA must converge to this unique NE.

To complete the proof, we just need to show that (\ref{(C0)}) is
equivalent to (\ref{SF_alpha}). \ Since in (\ref{SF_alpha}) each
$\alpha_{q}\in\lbrack0,1),$ we have\begin{equation}
\beta=\left\Vert \mathbf{S}_{\mathbf{\alpha}}^{\max}\right\Vert _{\infty}^{\mathbf{w}}<1\quad\Leftrightarrow\quad\left\Vert \mathbf{S}^{\max}\right\Vert _{\infty}^{\mathbf{w}}<1.\quad\label{SF_eq}\end{equation}
 Since $\mathbf{S}^{\max}$ is a nonnegative matrix, there exists
a positive vector $\overline{\mathbf{w}}$ such that \cite[Corollary 6.1]{Bertsekas Book-Parallel-Comp}\begin{equation}
\left\Vert \mathbf{S}^{\max}\right\Vert _{\infty}^{\overline{\mathbf{w}}}<1\quad\Leftrightarrow\quad\rho\left(\mathbf{S}^{\max}\right)<1.\quad\label{eq_nor_spectral_radius}\end{equation}
 Since the convergence of the asynchronous IWFA is guaranteed under
(\ref{SF_alpha}), for any given $\mathbf{w>0,}$ we can choose $\mathbf{w=}\overline{\mathbf{w}}$
and use (\ref{eq_nor_spectral_radius}).

Conditions (\ref{cond_C1})-(\ref{cond_C2}) in Corollary \ref{Corollary-C1-C2}
can be obtained as follows. Using \cite[Proposition 6.2e]{Bertsekas Book-Parallel-Comp}
%\begin{equation}
 $\rho(\mathbf{S}^{\max})\leq\Vert\mathbf{S}^{\max}\Vert_{\infty}^{\mathbf{w}},\forall\mathbf{w>0,}\,$%\end{equation}
 a sufficient condition for the $\Rightarrow$ direction in (\ref{eq_nor_spectral_radius})
is %\begin{equation}
$\Vert\mathbf{S}^{\max}\Vert_{\infty}^{\mathbf{w}}<1,$ %\quad
%\end{equation}
for some given $\mathbf{w>0;}$ which provides (\ref{cond_C1}). Condition
(\ref{cond_C2}) is obtained similarly, still using (\ref{eq_nor_spectral_radius})
and $\rho\left(\mathbf{S}^{\max}\right)=\rho\left(\mathbf{S}^{\max T}\right).$%TCIMACRO{\TeXButton{QED}{\hspace{\fill}\rule{1.5ex}{1.5ex}}}%
%BeginExpansion
\hspace{\fill}\rule{1.5ex}{1.5ex}%EndExpansion
\vspace{-0.1cm}

\section{Proof of Corollary \ref{Corollary:Conditions of the others}}

\vspace{-0.1cm}

\label{proof_Corollary:Conditions of the others} Since conditions (\ref{cond_C1})-(\ref{cond_C2}) imply  (\ref{(C0)}) (Corollary \ref{Corollary-C1-C2}), the sufficiency of
(\ref{Chung SF}) for (\ref{(C0)}) follows directly setting, in (\ref{cond_C1}),
$P_{q}=P_{r}$ for all $q,  r,$ $\mathcal{D}_{q}=\mathcal{D}_{r}=\{1,\ldots,N\},$
$\mathbf{w=1},$ and using the following upper bound $|\bar{H}_{rq}(k)|^{2}/|\bar{H}_{qq}(k)|^{2}\leq\max\limits _{r\neq q}\left\{ |\bar{H}_{rq}(k)|^{2}/|\bar{H}_{qq}(k)|^{2}\right\} $. Observe that  condition (\ref{Luo-Yamashitay}) is stronger than (\ref{Chung SF}), and thus implies (\ref{(C0)}).

We prove now that condition (\ref{Luo_Cond}) is stronger than  (\ref{(C0)}).
To this end, it is sufficient to show that $\rho\left(\mathbf{\Upsilon}\right)<1$
$\Leftrightarrow\rho\left(\overline{\mathbf{S}}^{\max}\right)<1,$ where $\overline{\mathbf{S}}^{\max}$ is defined after (\ref{Gamma_matrix}),
since $\mathbf{S}^{\max}\leq\overline{\mathbf{S}}^{\max}$ leads to
$\rho\left(\mathbf{S}^{\max}\right)\leq\rho\left(\overline{\mathbf{S}}^{\max}\right)<1$ \cite[Corollary 2.2.22]{CPStone92}.\footnote{The inequality $\mathbf{S}^{\max}\leq\overline{\mathbf{S}}^{\max}$ has to be intended componentwise.}
We first introduce the following intermediate definition and result
\cite{CPStone92}. \begin{definition} \label{Def:K matrix}A matrix
$\mathbf{A}\in\mathbb{R}^{n\times n}$ is said to be a $\mathbf{Z}$-matrix
if its off-diagonal entries are all non-positive. A matrix $\mathbf{A}\in\mathbb{R}^{n\times n}$
is said to be a $\mathbf{P}$-matrix if all its principal minors are
positive. A $\mathbf{Z}$-matrix that is also $\mathbf{P}$ is called
a $\mathbf{K}$-matrix. \end{definition}

\begin{lemma} [{\cite[Lemma $5.3.14$]{CPStone92}}]\label{Lemma-comparison-matrix}Let
$\mathbf{A}\in\mathbb{R}^{n\times n}$ be a $\mathbf{K}$-matrix and
$\mathbf{B}$ a nonnegative matrix. Then $\rho(\mathbf{A}^{-1}\mathbf{B})<1$
if and only if $\mathbf{A-B}$ is a $\mathbf{K}$-matrix. \end{lemma}

According to Definition \ref{Def:K matrix}, $\mathbf{I-}\overline{\mathbf{S}}_{\text{low}}^{\max}$
is a $\mathbf{Z}$-matrix. Since all principal minors of $\mathbf{I-}\overline{\mathbf{S}}_{\text{low}}^{\max}$
are equal to one (recall that $\mathbf{I-}\overline{\mathbf{S}}_{\text{low}}^{\max}$ is a lower triangular matrix with all ones on the main diagonal), $\mathbf{I-}\overline{\mathbf{S}}_{\text{low}}^{\max}$
is also a $\mathbf{P}$-matrix, and thus a $\mathbf{K}$-matrix. Invoking
Lemma \ref{Lemma-comparison-matrix} we obtain the following chain
of equivalences %\vspace{-0.3cm}
 \begin{equation}
\rho\left(\mathbf{\Upsilon}\right)<1\quad\Leftrightarrow\quad\mathbf{I}-\overline{\mathbf{S}}^{\max}\text{ \ is a }\mathbf{K}\text{-matrix}\quad\Leftrightarrow\quad\rho\left(\overline{\mathbf{S}}^{\max}\right)<1,\label{equiv_1}\end{equation}
 where the first and the second equivalence follows from Lemma \ref{Lemma-comparison-matrix}
using the correspondences $\mathbf{A}=\mathbf{I-}\overline{\mathbf{S}}_{\text{low}}^{\max}$,
$\mathbf{B}=\overline{\mathbf{S}}_{\text{upp}}^{\max}$, $\mathbf{I}-\overline{\mathbf{S}}_{\text{low}}^{\max}-\overline{\mathbf{S}}_{\text{upp}}^{\max}=\mathbf{I}-\overline{\mathbf{S}}^{\max}$
and $\mathbf{A}=\mathbf{I}$, $\mathbf{B}=\overline{\mathbf{S}}^{\max}$,
respectively. It follows from (\ref{equiv_1}) that $\rho\left(\mathbf{\Upsilon}\right)<1\Leftrightarrow\rho\left(\overline{\mathbf{S}}^{\max}\right)<1\Rightarrow\rho\left({\mathbf{S}}^{\max}\right)<1$;
which completes the proof.\hspace{\fill}\rule{1.5ex}{1.5ex}

\def\baselinestretch{0.9}
%\small
%\scriptsize

\end{document}